\pdfoutput=1
\documentclass[acmsmall,screen]{acmart}

\AtBeginDocument{%
  \providecommand\BibTeX{{%
    \normalfont B\kern-0.5em{\scshape i\kern-0.25em b}\kern-0.8em\TeX}}}

\setcopyright{rightsretained}
\acmPrice{}
\acmDOI{10.1145/3547641}
\acmYear{2022}
\copyrightyear{2022}
\acmSubmissionID{icfp22main-p52-p}
\acmJournal{PACMPL}
\acmVolume{6}
\acmNumber{ICFP}
\acmArticle{110}
\acmMonth{8}

\usepackage{xcolor}
\usepackage{mathpartir}
\usepackage{todonotes}
\presetkeys{todonotes}{inline}{}
\usepackage{scalerel}
\usepackage{bm}

\newcommand{\mit}[1]{\mathit{#1}}
\newcommand{\msf}[1]{\mathsf{#1}}
\newcommand{\mbb}[1]{\mathbb{#1}}

\newcommand{\bs}[1]{\boldsymbol{#1}}
\newcommand{\wh}[1]{\widehat{#1}}
\newcommand{\ext}{\triangleright}
\newcommand{\Code}{\msf{Code}}
\newcommand{\El}{\msf{El}}
\newcommand{\lam}{\msf{lam}}
\newcommand{\app}{\msf{app}}
\newcommand{\NatElim}{\msf{NatElim}}

\newcommand{\Lift}{{\Uparrow}}
\newcommand{\spl}{{\sim}}
\newcommand{\qut}[1]{\langle #1\rangle}

\newcommand{\mbbc}{\mbb{C}}
\newcommand{\mbbo}{\mbb{O}}
\newcommand{\ob}{_\mbbo}

\renewcommand{\U}{\msf{U}}
\newcommand{\Con}{\msf{Con}}
\newcommand{\Sub}{\msf{Sub}}
\newcommand{\Ty}{\msf{Ty}}
\newcommand{\Tm}{\msf{Tm}}
\newcommand{\Cono}{\msf{Con}_{\mbbo}}
\newcommand{\Subo}{\msf{Sub}_{\mbbo}}

\newcommand{\hCon}{\wh{\msf{Con}}}
\newcommand{\hSub}{\wh{\msf{Sub}}}
\newcommand{\hTy}{\wh{\msf{Ty}}}
\newcommand{\hTm}{\wh{\msf{Tm}}}

\newcommand{\p}{\mathsf{p}}
\newcommand{\q}{\mathsf{q}}

\newcommand{\refl}{\msf{refl}}
\newcommand{\Bool}{\msf{Bool}}
\newcommand{\true}{\msf{true}}
\newcommand{\false}{\msf{false}}
\newcommand{\True}{\msf{True}}
\newcommand{\False}{\msf{False}}
\newcommand{\List}{\msf{List}}
\newcommand{\nil}{\msf{nil}}
\newcommand{\cons}{\msf{cons}}
\newcommand{\Nat}{\msf{Nat}}
\newcommand{\zero}{\msf{zero}}
\newcommand{\suc}{\msf{suc}}
\renewcommand{\tt}{\msf{tt}}
\newcommand{\fst}{\msf{fst}}
\newcommand{\snd}{\msf{snd}}
\newcommand{\mylet}{\msf{let}}
\newcommand{\emptycon}{\scaleobj{.75}\bullet}
\newcommand{\id}{\msf{id}}

\newcommand{\Set}{\mathsf{Set}}
\newcommand{\Prop}{\mathsf{Prop}}
\newcommand{\Rep}{\msf{Rep}}
\newcommand{\blank}{{\mathord{\hspace{1pt}\text{--}\hspace{1pt}}}}
\newcommand{\emb}[1]{\ulcorner#1\urcorner}

\newcommand{\Stage}{\msf{Stage}}
\newcommand{\hato}{\bm\hat{\mbbo}}
\newcommand{\ev}{\mbb{E}}
\newcommand{\re}{\mbb{R}}

\theoremstyle{remark}
\newtheorem{notation}{Notation}

\newcommand{\whset}{\wh{\Set}}
\newcommand{\rexti}{\re_{\ext_1}^{-1}}
\newcommand{\rextizero}{\re_{\ext_0}^{-1}}

\newcommand{\rel}{^{\approx}}
\newcommand{\yon}{\msf{y}}

\begin{document}

\title{Staged Compilation with Two-Level Type Theory}
\titlenote{The author was supported by the “Application Domain Specific Highly
  Reliable IT Solutions” project which has been implemented with support from
  the National Research, Development and Innovation Fund of Hungary, financed
  under the Thematic Excellence Programme TKP2020-NKA-06 (National Challenges
  Subprogramme) funding scheme.}

\author{András Kovács}
\email{kovacsandras@inf.elte.hu}
\orcid{0000-0002-6375-9781}
\affiliation{%
  \institution{Eötvös Loránd University}
  \country{Hungary}
  \city{Budapest}
}

\begin{abstract}
The aim of staged compilation is to enable metaprogramming in a way such that we
have guarantees about the well-formedness of code output, and we can also mix
together object-level and meta-level code in a concise and convenient manner. In
this work, we observe that two-level type theory (2LTT), a system originally
devised for the purpose of developing synthetic homotopy theory, also serves as
a system for staged compilation with dependent types. 2LTT has numerous good
properties for this use case: it has a concise specification, well-behaved model
theory, and it supports a wide range of language features both at the object and
the meta level. First, we give an overview of 2LTT's features and applications
in staging. Then, we present a staging algorithm and prove its correctness. Our
algorithm is ``staging-by-evaluation'', analogously to the technique of
normalization-by-evaluation, in that staging is given by the evaluation of 2LTT
syntax in a semantic domain. The staging algorithm together with its correctness
constitutes a proof of strong conservativity of 2LLT over the object theory. To our
knowledge, this is the first description of staged compilation which supports
full dependent types and unrestricted staging for types.
\end{abstract}



\maketitle

\section{Introduction}\label{sec:introduction}

The purpose of staged compilation is to write code-generating programs in a
safe, ergonomic and expressive way. It is always possible to do ad-hoc code
generation, by simply manipulating strings or syntax trees in a sufficiently
expressive programming language. However, these approaches tend to suffer from
verbosity, non-reusability and lack of safety. In staged compilation, there are
certain \emph{restrictions} on which metaprograms are expressible. Usually,
staged systems enforce typing discipline, prohibit arbitrary manipulation of
object-level scopes, and often they also prohibit accessing the internal
structure of object expressions. On the other hand, we get \emph{guarantees}
about the well-scoping or well-typing of the code output, and we are also able
to use concise syntax for embedding object-level code.

\emph{Two-level type theory}, or 2LTT in short, was described in
\cite{twolevel}, building on ideas from \cite{hts}. The motivation was to allow
convenient metatheoretical reasoning about a certain mathematical language
(homotopy type theory), and to enable concise and modular ways to extend the
language with axioms.

It turns out that metamathematical convenience closely corresponds to
metaprogramming convenience: 2LTT can be directly and effectively employed in
staged compilation. Moreover, semantic ideas underlying 2LTT are also directly
applicable to the theory of staging.

\subsection{Overview \& Contributions}\label{sec:overview}

\begin{itemize}
  \item In Section \ref{sec:tour-of-2ltt} we present an informal syntax of
    two-level type theory, a dependent type theory with support for two-stage
    compilation. We look at basic use cases involving inlining control, partial
    evaluation and fusion optimizations. We also describe feature variations,
    enabling applications in monomorphization and memory layout control.
  \item In Section \ref{sec:formal-2ltt}, we specify the formal notion of models
    and syntax for the variant of 2LTT used in this paper. We mostly work with a
    high-level algebraic presentation, where the syntax is quotiented by
    conversion. However, we also explain how to extract functions
    operating on non-quotiented syntax, by interpreting type-theoretic
    constructions in a setoid model.
  \item
    In Section \ref{sec:staging-algorithm}, we review the standard presheaf
    model of 2LTT \cite[Section~2.5.3]{twolevel}, which lies over the syntactic
    category of the object theory. We show that evaluation in this model yields
    a staging algorithm for closed types and terms. We then extend staging to
    open types and terms which may depend on object-level variables. We show
    stability of staging, which means that staging is surjective up to
    conversion, and we also show that staging strictly preserves all type and
    term formers. Finally, we discuss efficiency and potential practical implementations
    of staging.
  \item
    In Section \ref{sec:soundness} we show soundness of staging, which roughly
    means that the output of staging is convertible to its input. Staging
    together with its stability and soundness can be viewed as a \emph{strong
    conservativity} theorem of 2LTT over the object theory. This means that the
    same constructions can be expressed in the object theory and the
    object-level fragment of 2LTT, up to conversion, and staging witnesses that
    meta-level constructions can be always computed away. This strengthens the
    weak notion of conservativity shown in \cite{capriotti2017models} and
    \cite{twolevel}.
  \item
    In Section \ref{sec:intensional-analysis}, we discuss possible semantic
    interpretations of intensional code analysis, i.e.\ the ability to look into
    the structure of object-level code.
  \item We provide a small standalone implementation of elaboration and staging
    for a two-level type theory \cite{staged-demo}.
  \item To our knowledge, this work is the first formalization and implementation of
    staged compilation in the presence of full dependent types, with
    universes and large elimination. In particular, we allow unrestricted
    staging for types, so that types can be computed by metaprograms at compile
    time.
\end{itemize}

\section{A Tour of Two-Level Type Theory}\label{sec:tour-of-2ltt}

In this section, we provide a short overview of 2LTT and its potential
applications in staging. We work in the informal syntax of a dependently typed
language which resembles Agda \cite{agdadocs}. We focus on examples and informal
explanation here; the formal details will be presented in Section \ref{sec:formal-2ltt}.

\begin{notation}\label{basic-notation}
We use the following notations throughout the paper. $(x : A) \to B$ denotes a
dependent function type, where $x$ may occur in $B$. We use $\lambda\,x.\,t$ for
abstraction. A $\Sigma$-type is written as $(x : A) \times B$, with pairing as
$(t,\,u)$, projections as $\fst$ and $\snd$, and we may use pattern matching
notation on pairs, e.g.\ as in $\lambda\,(x,\,y).\,t$. The unit type is $\top$
with element $\tt$. We will also use Agda-style notation for implicit arguments,
where $t : \{x : A\} \to B$ implies that the first argument to $t$ is inferred
by default, and we can override this by writing a $t \{u\}$ application. We may
also implicitly quantify over arguments (in the style of Idris and Haskell), for
example when declaring $\mit{id} : A \to A$ with the assumption that $A$ is
universally quantified.
\end{notation}

\subsection{Rules of 2LTT}

\subsubsection{Universes}
We have universes $\U_{i,j}$, where $i \in \{0,1\}$, and $j \in \mbb{N}$.  The
$i$ index denotes stages, where $0$ is the runtime (object-level) stage, and $1$
is the compile time (meta-level) stage. The $j$ index denotes universe sizes in
the usual sense of type theory. We assume Russell-style universes, with
$\U_{i,j} : \U_{i, j+1}$. However, for the sake of brevity we will usually omit
the $j$ indices in this section, as sizing is orthogonal to our use-cases and
examples.
\begin{itemize}
\item $\U_0$ can be viewed as the \emph{universe of object-level or runtime types}.
    Each closed type $A : \U_0$ can be staged to an actual type in the object language
    (the language of the staging output).
  \item $\U_1$ can be viewed as the \emph{universe of meta-level or static types}. If we
    have $A : \U_1$, then $A$ is guaranteed to be only present at compile time,
    and will be staged away. Elements of $A$ are likewise computed away.
\end{itemize}

\subsubsection{Type Formers} $\U_0$ and $\U_1$ may be closed under arbitrary type formers,
such as functions, $\Sigma$-types, identity types or inductive types in general.
However, all constructors and eliminators in type formers must stay at the same
stage. For example:
\begin{itemize}
  \item Function domain and codomain types must be at the same stage.
  \item If we have $\Nat_0 : \U_0$ for the runtime type of natural numbers,
        we can only map from it to a type in $\U_0$ by recursion or induction.
\end{itemize}
It is not required that we have the \emph{same} type formers at both stages. We
will discuss variations of the object-level languages in Section \ref{sec:variations}.

\subsubsection{Moving Between Stages}
At this point, our system is rather limited, since there is no interaction
between the stages. We enable such interaction through the following operations.
Note that none of these operations can be expressed as functions, since function
types cannot cross between stages.
\begin{itemize}
\item \emph{Lifting:} for $A : \U_0$, we have $\Lift A : \U_1$.  From the
  staging point of view, $\Lift A$ is the type of metaprograms which compute
  runtime expressions of type $A$.
\item \emph{Quoting:} for $A : \U_0$ and $t : A$, we have $\qut{t} :\,\Lift A$.
  A quoted term $\qut{t}$ represents the metaprogram which immediately yields
  $t$.
\item \emph{Splicing:} for $A : \U_0$ and $t :\,\Lift A$, we have $\spl t : A$.
  During staging, the metaprogram in the splice is executed, and the resulting
  expression is inserted into the output.

  \begin{notation} Splicing binds stronger than any operation, including function
    application. For instance, $\spl f\,x$ is parsed as $(\spl f)\,x$. We borrow
    this notation from MetaML \cite{metaml}.
  \end{notation}

\item Quoting and splicing are definitional inverses, i.e.\ we have $\spl\qut{t} = t$ and
  $\qut{\spl t}=t$ as definitional equalities.
\end{itemize}

Informally, if we have a closed program $t : A$ with $A : \U_0$, \emph{staging}
means computing all metaprograms and recursively replacing all splices in $t$ and $A$
with the resulting runtime expressions. The rules of 2LTT ensure that this is possible,
and we always get a splice-free object program after staging.

\paragraph{Remark}
Why do we use the index $0$ for the runtime stage? The reason is that it is not
difficult to generalize 2LTT to multi-level type theory, by allowing to lift
types from $\U_i$ to $\U_{i+1}$. In the semantics, this can be modeled by
having a 2LTT whose object theory is once again a 2LTT, and doing this in an
iterated fashion. But there must be necessarily a bottom-most object theory;
hence our stage indexing scheme. For now though, we leave the multi-level
generalization to future work.

\begin{notation}\label{notation:basics} We may disambiguate type formers at different stages by using $0$ or $1$
subscripts. For example, $\Nat_1 : \U_1$ is distinguished from $\Nat_0 : \U_0$,
and likewise we may write $\zero_0 : \Nat_0$ and so on. For function and
$\Sigma$ types, the stage is usually easy to infer, so we do not annotate
them. For example, the type $\Nat_0 \to \Nat_0$ must be at the runtime stage,
since the domain and codomain types are at that stage, and we know that the
function type former stays within a single stage. We may also omit stage annotations
from $\lambda$ and pairing.
\end{notation}

\subsection{Staged Programming in 2LTT}\label{sec:staged-programming}

In 2LTT, we may have several different polymorphic identity functions. First,
consider the usual identity function at each stage:
\begin{alignat*}{3}
  & \mit{id}_0 : (A : \U_0) \to A \to A\hspace{2em} && \mit{id}_1 : (A : \U_1) \to A \to A\\
  & \mit{id}_0 := \lambda\,A\,x.x       && \mit{id}_1 := \lambda\,A\,x.x
\end{alignat*}
An $\mit{id}_0$ application will simply appear in staging output as it is. In
contrast, $\mit{id}_1$ can be used as a compile-time evaluated function, because
the staging operations allow us to freely apply $\mit{id}_1$ to runtime
arguments. For example, $\mit{id}_1\,(\Lift\,\Bool_0)\,\qut{\true_0}$ has type
$\Lift \Bool_0$, which implies that $\spl(\mit{id}_1\,(\Lift\,\Bool_0)\,\qut{\true_0})$
has type $\Bool_0$. We can stage this expression as follows:
\[
\spl(\mit{id}_1\,(\Lift\,\Bool_0)\,\qut{\true_0}) = \spl\qut{\true_0} = \true_0
\]
There is another identity function, which computes at compile time, but which
can be only used on runtime arguments:
\begin{alignat*}{3}
  & \mit{id_\Lift} : (A : \Lift\U_0) \to \Lift\,\spl A \to \Lift\,\spl A\\
  & \mit{id_\Lift} := \lambda\,A\,x.x
\end{alignat*}
Note that since $A : \Lift\U_0$, we have $\spl A : \U_0$, hence $\Lift\,\spl A
: \U_1$.  Also, $\Lift\U_0 : \U_1$, so all function domain and codomain types in
the type of $\mit{id_\Lift}$ are at the same stage. Now, we may write
$\spl(\mit{id_\Lift}\,\qut{\Bool_0}\,\qut{\true_0})$ for a term which is staged
to $\true_0$. In this specific case $\mit{id_\Lift}$ has no practical advantage
over $\mit{id}_1$, but in some cases we really have to quantify over
$\Lift\U_0$. This brings us to the next example.

Assume $\List_0 : \U_0 \to \U_0$ with $\nil_0 : (A : \U_0) \to \List_0\,A$,
$\cons_0 : (A : \U_0) \to A \to \List_0\,A$ and $\msf{foldr}_0 : (A\,B : \U_0)
\to (A \to B \to B) \to B \to \List_0\,A \to B$. We define a mapping function which
inlines its function argument:
\begin{alignat*}{3}
  & \mit{map} : (A\,B : \Lift\U_0) \to (\Lift\,\spl A \to \Lift\,\spl B)
      \to \Lift(\List_0\,\spl A) \to \Lift(\List_0\,\spl B)\\
  & \mit{map} := \lambda\,A\,B\,f\,\mit{as}.\,
      \qut{\msf{foldr}_0\,
        \spl A\,(\List_0 \spl B)\,
        (\lambda\,a\,\mit{bs}.\, \cons_0\,\spl B\,\spl(f\,\qut{a})\,\mit{bs})\,
        (\nil_0\,\spl B)\,
        \spl as
        }
\end{alignat*}
This $\mit{map}$ function can be defined with quantification over $\Lift \U_0$ but not
over $\U_1$, because $\List_0$ expects type parameters in $\U_0$, and there is
no generic way to convert from $\U_1$ to $\U_0$. Now, assuming
$\blank\!+_0\!\blank : \Nat_0 \to \Nat_0 \to \Nat_0$ and $\mit{ns} :
\List_0\,\Nat_0$, we have the following staging behavior:
\begin{alignat*}{3}
  & &&\spl(\mit{map}\,\qut{\Nat_0}\,\qut{\Nat_0}\,(\lambda\,n.\,\qut{\spl n +_0 10})\,\qut{\mit{ns}})\\
  & =\hspace{0.3em}&&\spl\qut{\msf{foldr}_0\,
        \spl \qut{\Nat_0}\,(\List_0\,\spl\qut{\Nat_0})\\
  & && \hspace{3.5em}(\lambda\,a\,\mit{bs}.\, \cons_0\,\spl\qut{\Nat_0}\,\spl\qut{\spl\qut{a} +_0 10}\,\mit{bs})\,
        (\nil_0\,\spl \qut{\Nat_0})\,
        \spl\qut{\mit{ns}}}\\
  & =\hspace{0.3em} &&\msf{foldr_0}\,\Nat_0\,(\List_0\,\Nat_0)\,\,(\lambda\,a\,bs.\,\cons_0\,\Nat_0\,(a +_0 10)\,bs)\,(\nil_0\,\Nat_0)\,\mit{ns}
\end{alignat*}

By using meta-level functions and lifted types, we already have control over
inlining. However, if we want to do more complicated meta-level computation, it
is convenient to use recursion or induction on meta-level type formers. A
classic example in staged compilation is the power function for natural numbers,
which evaluates the exponent at compile time. We assume the iterator function
$\msf{iter_1} : \{A : \U_1\} \to \Nat_1 \to (A \to A) \to A \to A$, and runtime
multiplication as $\blank\!*_0\!\blank$.
\begin{alignat*}{3}
  &\mit{exp} : \Nat_1 \to \Lift \Nat_0 \to \Lift \Nat_0 \\
  &\mit{exp} := \lambda\,x\,y.\,
  \msf{iter}_1\,x\,(\lambda\,n.\,\qut{\spl y *_0 \spl n})\,\qut{1}
\end{alignat*}
Now, $\spl(\mit{exp}\,3\,\qut{n})$ stages to $n *_0 n *_0 n *_0 1$ by the computation rules
of $\msf{iter_1}$ and the staging operations.

We can also stage \emph{types}. Below, we use iteration to compute the type of
vectors with static length, as a nested pair type.
\begin{alignat*}{3}
  &\mit{Vec} : \Nat_1 \to \Lift \U_0 \to \Lift \U_0\\
  &\mit{Vec} := \lambda\,n\,A.\,\msf{iter}_1\,n\,(\lambda\,B.\,\qut{\spl A \times \spl B})\,\qut{\top_0}
\end{alignat*}
With this definition, $\spl(\mit{Vec}\,3\,\qut{\Nat_0})$ stages to $\Nat_0
\times (\Nat_0 \times (\Nat_0 \times \top_0))$. Now, we can use \emph{induction}
on $\Nat_1$ to implement a map function. For readability, we use an Agda-style
pattern matching definition below (instead of the elimination principle).
\begin{alignat*}{6}
  &\hspace{-3em}\rlap{$\mit{map} : (n : \Nat_1) \to (\Lift\,\spl A \to \Lift\,\spl B) \to \Lift(\mit{Vec}\,n\,A) \to \Lift(\mit{Vec}\,n\,B)$}\\
  &\hspace{-3em}\mit{map}\,\,&&\zero_1      && f\,&&\mit{as} := \qut{\tt_0}\\
  &\hspace{-3em}\mit{map}\,\,&&(\suc_1\,n)\,&& f\,&&\mit{as} :=
     \qut{(\spl(f\,\qut{\fst_0\,\spl\mit{as}}),\,\spl(\mit{map}\,n\,f\,\qut{\snd_0\,\spl \mit{as}}))}
\end{alignat*}
This definition inlines the mapping function for each projected element of the
vector. For instance, staging $\spl(\mit{map}\,2\,(\lambda\,n.\,\qut{\spl n +_0
  10})\,\qut{\mit{ns}})$ yields $(\fst_0\,\mit{ns} +_0
10,\,(\fst_0(\snd_0\,\mit{ns}) +_0 10,\,\tt_0))$. Sometimes, we do not want to
duplicate the code of the mapping function. In such cases, we can use
\emph{let-insertion}, a common technique in staged compilation. If we bind a runtime
expression to a runtime variable, and only use that variable in subsequent
staging, only the variable itself can be duplicated. One solution is to do an
ad-hoc let-insertion:
\begin{alignat*}{6}
  &   && \mylet_0\,f := \lambda\,n.\, n +_0 10\,\,\msf{in}\,\,
         \spl(\mit{map}\,2\,(\lambda\,n.\,\qut{f\,\spl n})\,\qut{\mit{ns}}) \\
  & =\,\,&&  \mylet_0\,f := \lambda\,n.\, n +_0 10\,\,\msf{in}\,\,
          (f\,(\fst_0\,\mit{ns}),\,(f\,(\fst_0(\snd_0\,\mit{ns})),\,\tt_0))
\end{alignat*}
We include examples of more sophisticated let-insertion in the supplementary
code \cite{staged-demo}.

More generally, we are free to use dependent types at the meta-level, so we can
reproduce more complicated staging examples. Any well-typed interpreter can be
rephrased as a \emph{partial evaluator}, as long as we have sufficient type
formers. For instance, we may write a partial evaluator for a simply typed
lambda calculus. We sketch the implementation in the following; the full version
can be found in the supplementary code. First, we inductively define types,
contexts and terms:
\begin{alignat*}{6}
  & \Ty : \U_1  \hspace{2em} \Con : \U_1 \hspace{2em} \Tm : \Con \to \Ty \to \U_1
\end{alignat*}
Then we define the interpretation functions:
\begin{alignat*}{6}
  & \msf{EvalTy}  &&: \Ty \to \Lift \U_0 \\
  & \msf{EvalCon} &&: \Con \to \U_1 \\
  & \msf{EvalTm}  &&: \Tm\,\Gamma\,A \to \msf{EvalCon}\,\Gamma \to \Lift\,\spl(\msf{EvalTy}\,A)
\end{alignat*}
Types are necessarily computed to runtime types, e.g.\ an embedded
representation of the natural number type is evaluated to $\qut{\Nat_0}$.
Contexts are computed as follows:
\begin{alignat*}{4}
  &\msf{EvalCon}\,\,\msf{empty}                &&:= \top_1 \\
  &\msf{EvalCon}\,\,(\msf{extend}\,\Gamma\,A) &&:= \msf{EvalCon}\,\Gamma \times (\Lift \spl(\msf{EvalTy}\,A))
\end{alignat*}
This is an example for the usage of \emph{partially static data} \cite{partial-evaluation}:
semantic contexts are \emph{static} lists storing \emph{runtime}
expressions. This allows us to completely eliminate environment lookups in the
staging output: an embedded lambda expression is staged to the corresponding
lambda expression in the object language. This is similar to the partial
evaluator presented in Idris 1 \cite{scrap-your-inefficient-engine}. However, in
contrast to 2LTT, Idris 1 does not provide a formal guarantee that partial
evaluation does not get stuck.

\subsection{Properties of Lifting, Binding Time Improvements}

We describe more generally the action of lifting on type formers. While lifting
does not have definitional computation rules, it does preserve negative type formers up
to definitional isomorphism \cite[Section~2.3]{twolevel}:
\begin{alignat*}{5}
  \Lift((x : A) \to B\,x) &\simeq ((x : \Lift A) \to \Lift\,(B\,\spl x))\\
  \Lift ((x : A) \times B\,x) &\simeq ((x : \Lift A) \times \Lift (B\,\spl x))\\
  \Lift \top_0 &\simeq \top_1
\end{alignat*}
For function types, the preservation maps are the following:
\begin{alignat*}{5}
  & \hspace{-6em}\rlap{$\mit{pres}_\to : \Lift((x : A) \to B\,x) \to ((x : \Lift A) \to \Lift\,(B\,\spl x))$}\\
  & \hspace{-6em}\mit{pres}_\to\,f     &&:= \lambda\,x.\,\,\qut{\spl f\,\spl x}\\
  & \hspace{-6em}\mit{pres}_\to^{-1}\,f &&:= \qut{\lambda\,x.\,\,\spl(f\,\qut{x})}
\end{alignat*}
With this, we have that $\mit{pres}_\to\,(\mit{pres}_\to^{-1}\,f)$ is
definitionally equal to $f$, and also the other way around. Preservation maps
for $\Sigma$ and $\top$ work analogously.

By rewriting a 2LTT program left-to-right along preservation maps, we perform
what is termed \emph{binding time improvement} in the partial evaluation
literature \cite[Chapter~12]{partial-evaluation}. Note that the output of $\mit{pres}_{\to}$ uses a
meta-level $\lambda$, while going the other way introduces a runtime
binder. Meta-level function types and $\Sigma$-types support more computation during
staging, so in many cases it is beneficial to use the improved forms. In some
cases though we may want to use unimproved forms, to limit the size of generated
code. This is similar to what we have seen with let-insertion. For a minimal
example, consider the following unimproved version of $\mit{id}_\Lift$:
\begin{alignat*}{5}
  & \mit{id}_\Lift : (A : \Lift \U_0) \to \Lift(\spl A \to \spl A) \\
  & \mit{id}_\Lift := \lambda\,A.\,\,\qut{\lambda\,x.\,x}
\end{alignat*}
This can be used at the runtime stage as
$\spl(\mit{id}_\Lift\,\qut{\Bool_0})\,\true_0$, which is staged to
$(\lambda\,x.\,x)\,\true_0$.  This introduces a useless $\beta$-redex, so in
this case the improved version is clearly preferable.

For inductive types in general we do not get full preservation, only maps in one
direction. For example, we have $\Bool_1 \to \Lift \Bool_0$, defined as
$\lambda\,b.\,\msf{if}\,b\,\msf{then}\,\qut{\true_0}\,\msf{then}\,\qut{\false_0}$.
In the staging literature, this is called ``serialization'' or ``lifting''
\cite{template-haskell,metaml}.  In the other direction, we can only define
constant functions from $\Lift \Bool_0$ to $\Bool_1$.

The lack of elimination principles for $\Lift A$ means that we cannot inspect
the internal structure of expressions. We will briefly discuss ways to lift this
restriction in Section \ref{sec:intensional-analysis}.

In particular, we have no serialization map from $\Nat_1 \to \Nat_1$ to
$\Lift(\Nat_0 \to \Nat_0)$. However, when $A : \U_1$ is \emph{finite}, and $B :
\U_1$ can be serialized, then $A \to B$ can be serialized, because it is
equivalent to a finite product. For instance, $\Bool_1 \to \Nat_1 \simeq \Nat_1
\times \Nat_1$.  In 2LTT, $A$ is called \emph{cofibrant}
\cite[Section~3.4]{twolevel}: this means that for each $B$, $A \to \Lift\,B$ is
equivalent to $\Lift C$ for some $C$. This is the 2LTT formalization of the
so-called ``trick'' in partial evaluation, which improves binding times by
$\eta$-expanding functions out of finite sums \cite{eta-expansion-trick}.

\subsubsection{Fusion}
Fusion optimizations can be viewed as binding time improvement techniques for
general inductive types. The basic idea is that by lambda-encoding an inductive
type, it is brought to a form which can be binding-time improved. For instance,
consider foldr-build fusion for lists \cite{short-cut}, which is employed in GHC
Haskell.  Starting from $\Lift (\List_0\,A)$, we use Böhm-Berarducci encoding
\cite{boehm-berarducci} under the lifting to get
\begin{alignat*}{3}
         &\Lift((L : \U_0) \to (A \to L \to L) \to L \to L)\\
  \simeq\,\,&((L : \Lift\,U_0) \to (\Lift\,A \to \Lift\,\spl L \to \Lift\,\spl L) \to \Lift\,\spl L \to \Lift\,\spl L).
\end{alignat*}
Alternatively, for \emph{stream fusion} \cite{stream-fusion}, we embed $\List\,A$ into the
coinductive colists (i.e.\ the possibly infinite lists), and use a terminal
lambda-encoding. The embedding into the ``larger'' structure enables some staged
optimizations which are otherwise not possible, such as fusion for the
$\msf{zip}$ function. However, the price we pay is that converting
back to lists from colists is not necessarily total.

We do not detail the implementation of fusion in 2LTT here; a small example can
be found in the supplementary code. In short, 2LTT is a natural setting for a
wide range of fusion setups. A major advantage of fusion in 2LTT is the formal
guarantee of staging, in contrast to implementations where compile-time
computation relies on ad-hoc user annotations and general-purpose optimization
passes. For instance, fusion in GHC relies on rewrite rules and inlining
annotations which have to be carefully tuned and ordered, and it is possible to
get pessimized code via failed fusion.

\subsubsection{Inferring Staging Operations}
We can extract a coercive subtyping system from the staging operations, which
can be used for inference, and in particular for automatically transporting
definitions along preservation isomorphisms. One choice is to have $A \leq \Lift
A$, $\Lift A \leq A$, a contravariant-covariant rule for functions and a
covariant rule for $\Sigma$. During bidirectional elaboration, when we need to
compare an inferred and an expected type, we can insert coercions.  We
implemented this feature in our prototype. It additionally supports Agda-style
implicit arguments and pattern unification, so it can elaborate the following
definition:
\begin{alignat*}{3}
  & \mit{map} : \{A\,B : \Lift\U_0\} \to (\Lift A \to \Lift B)
      \to \Lift(\List_0\,A) \to \Lift(\List_0\,B)\\
  & \mit{map} := \lambda\,f\,\mit{as}.\,
      \msf{foldr}_0
        (\lambda\,a\,\mit{bs}.\, \cons_0\,(f\,a)\,\mit{bs})\,
        \nil_0\,
        as
\end{alignat*}
We may go a bit further, and also add the coercive subtyping rule $\U_0 \leq
\U_1$, witnessed by $\Lift$. Then, the type of $\mit{map}$ can be written as
$\{A\,B : \Lift\U_0\} \to (A \to B) \to \List_0\,A \to \List_0\,B$. However,
here the elaborator has to make a choice, whether to elaborate to improved or
unimproved types. In this case, the fully unimproved type would be
\[ \{A\,B : \Lift\U_0\} \to \Lift((\spl A \to \spl B) \to \List_0\,\spl A \to \List_0\,\spl B). \]
It seems that improved types are a sensible default, and we can insert explicit
lifting when we want to opt for unimproved types. This is also available in our
prototype.

\subsection{Variations of Object-Level Languages}
\label{sec:variations}

In the following, we consider variations on object-level languages, with a focus
on applications in downstream compilation after staging. Adding restrictions or
more distinctions to the object language can make it easier to optimize and
compile.

\subsubsection{Monomorphization}\label{sec:monomorphization}

In this case, the object language is simply typed, so every type is known
statically. This makes it easy to assign different memory layouts to different
types and generate code accordingly for each type. Moving to 2LTT, we still
want to abstract over runtime types at compile time, so we use the following
setup.
\begin{itemize}
\item We have a \emph{jugdment}, written as $A\,\msf{type_0}$, for well-formed
  runtime types. Runtime types may be closed under simple type formers.
\item We have a type $\Ty_0 : \U_1$ in lieu of the previous $\Lift \U_0$.
\item For each $A\,\msf{type}_0$, we have $\Lift A : \Ty_0$.
\item We have quoting and splicing for types and terms. For types, we send $A\,
  \msf{type}_0$ to $\qut{A} : \Ty_0$. For terms, we send $t : A$ to $\qut{t} :
  \Lift A$.
\end{itemize}
Despite the restriction to simple types at runtime, we can still write arbitrary
higher-rank polymorphic functions in 2LTT, such as a function with type $((A : \Ty_0)
\to \Lift\,\spl A \to \Lift \spl A) \to \Lift\,\Bool_0$. This function can be
only applied to statically known arguments, so the polymorphism can be staged
away. The main restriction that programmers have to keep in mind is that
polymorphic functions cannot be stored inside runtime data types.


\subsubsection{Memory Representation Polymorphism}

This refines monomorphization, so that types are not directly identified with
memory representations, but instead representations are internalized in 2LTT as
a meta-level type, and runtime types are indexed over representations.
\begin{itemize}
\item We have $\Rep : \U_1$ as the type of memory representations. We have
  considerable freedom in the specification of $\Rep$. A simple setup may
  distinguish references from unboxed products, i.e.\ we have $\msf{Ref} : \Rep$
  and $\msf{Prod} : \Rep \to \Rep \to \Rep$, and additionally we may assume
  any desired primitive machine representation as a value of $\Rep$.
\item We have Russell-style $\U_{0,j} : \Rep \to \U_{0, j+1}\,r$, where $r$ is
  some chosen runtime representation for types; usually we would mark types are
  erased. We leave the meta-level $\U_{1,j}$ hierarchy unchanged.
\item We may introduce unboxed $\Sigma$-types and primitive machine types in the
  runtime language. For $r : \Rep$, $r' : \Rep$, $A : \U_{0}\,r$ and $B : A \to
  \U_{0}\,r'$, we may have $(x : A) \times B\,x :
  \U_{0}\,(\msf{Prod\,r\,r'})$. Thus, we have type dependency, but we do not
  have dependency in memory representations.
\end{itemize}
Since $\Rep$ is meta-level, there is no way to abstract over it at runtime, and
during staging all $\Rep$ indices are computed to concrete canonical
representations. This is a way to reconcile dependent types with some amount of
control over memory layouts. The unboxed flavor of $\Sigma$ ends up with a
statically known flat memory representation, computed from the representations
of the fields.


\section{Formal Setup}\label{sec:formal-2ltt}

In this section we describe our approach to formalizing staging and also the
\emph{extraction} of algorithms from the formalization. There is a tension
here:
\begin{itemize}
  \item On one hand, we would like to do formal work at a higher level of
    abstraction, in order to avoid a deluge of technical details (for which
    dependent type theories are infamous).
  \item On the other hand, at some point we need to talk about lower-level
    operational aspects and extracted algorithms.
\end{itemize}
In the following, we describe our approach to interfacing between the different levels
of abstraction.

\subsection{The Algebraic Setting}

When defining staging and proving its soundness and stability, we work
internally in a constructive type theory, using types and terms that are
quotiented by conversion. In this setting, every construction respects syntactic
conversion, which is very convenient technically. Here, we use standard
definitions and techniques from the categorical and algebraic metatheory of type
theories, and reuse prior results which rely on the same style of
formalization. Concretely, we use \emph{categories with families} \cite{cwfs}
which is a commonly used specification of an \emph{explicit substitution
calculus} of types and terms.

In the rest of the paper we will mostly work in this setting, and we will
explicitly note when we switch to a different metatheory. We summarize
the features of this metatheory in the following.

\paragraph{Metatheory}
The metatheory is an intensional type theory with a cumulative universe
hierarchy $\Set_i$, where $i$ is ordinal number such that $i \leq \omega +
1$. The transfinite $i$ levels are a convenience feature for modeling typing
contexts, where sometimes we will need universes large enough to fit types at
arbitrary finite levels. We also assume uniqueness of identity proofs and
function extensionality. We will omit transports along identity proofs in the
paper, for the sake of readability. The missing transports can be in principle
recovered due to Hofmann's conservativity result \cite{hofmann95extensional};
see also \cite{DBLP:conf/cpp/WinterhalterST19}. We assume $\Pi$, $\Sigma$,
$\top$ and $\mbb{N}$ types in all universes, and reuse Notation
\ref{notation:basics}. We also assume certain quotient inductive-inductive
types, to represent the syntax of 2LTT and the object theory. We will detail
these shortly in Sections \ref{sec:2lttmod} and \ref{sec:objmod}.

\subsection{Algorithm Extraction}

In the algebraic setting we can \emph{only} talk about properties which
respect syntactic conversion. However, the main motivation of staged compilation is to
improve performance and code abstraction, but these notions do \emph{not}
respect conversion. For instance, rewriting programs along $\beta$-conversion
can improve performance. Hence, at some point we may want to ``escape'' from
quotients.

Fortunately, our specified meta type theory can be interpreted in
\emph{setoids}, as most recently explained in
\cite{DBLP:journals/pacmpl/PujetT22}. A setoid is a set together with an
equivalence relation on it. We interpret closed types as setoids, and dependent
types as setoid fibrations. For inductive base types such $\Bool$ or $\Nat$, the
equivalence relation is interpreted as identity. Functions are interpreted as
equivalence-preserving functions, and two functions are equivalent if they are
pointwise equivalent. A quotient type $A/R$ is interpreted by extending the
semantic notion of equivalence for $A$ with the given $R$ relation. We will
sketch the interpretation of the 2LTT syntax (given as a quotient
inductive-inductive type) in Section \ref{sec:2lttsetoid}.

In particular, from $f : A/R \to A/R$ we can extract a function which acts on
the underlying set of $A$'s interpretation. Note that an extracted function can
be only viewed as an algorithm if the setoid interpretation is itself defined in
a constructive theory, but that is the case in the cited work
\cite{DBLP:journals/pacmpl/PujetT22}.

\begin{definition} Assume a closed function $f : A \to B$ defined
in the mentioned type theory. The \textbf{extraction} of $f$ is
the underlying function of $f$'s interpretation in the setoid model.
\end{definition}

In other words, we can view quotients as a convenient shorthand for working with
explicit equivalence relations. We can work with quotients when we want to enforce
preservation of equivalence relations, but at any point we have the option to
switch to an external view and look at underlying sets and functions.


\subsection{Models and Syntax of 2LTT}\label{sec:2lttmod}

We specify the notion of model for 2LTT in this section. The syntax of 2LTT will
be defined as the initial model, as a quotient inductive-inductive type,
following \cite{ttintt}. This yields an explicit substitution calculus,
quotiented by definitional equality. A major advantage of this representation is
that it \emph{only} includes well-typed syntax, so we never have to talk about
``raw'' types and terms.

We provide an appendix to this paper, where we present a listing of all rules
and equations of 2LTT in two styles: first, in a more traditional style with
derivation rules, and also in a very compact manner using a second-order
algebraic signature.

First, we define the structural scaffolding of 2LTT without type formers, mostly
following \cite{twolevel}. The specification is agnostic about the sizes of sets
involved, i.e.\ we can model underlying sorts using any metatheoretic $\Set_j$.

\begin{definition} A model of \textbf{basic 2LTT} consists of the following.
\begin{itemize}
\item
  A category $\mbbc$ with a terminal object. We denote the set of objects as
  $\Con_{\mbbc} : \Set$ and use capital Greek letters starting from $\Gamma$ to
  refer to objects. The set of morphisms is $\Sub_{\mbbc} : \Con_{\mbbc} \to
  \Con_{\mbbc} \to \Set$, and we use $\sigma$, $\delta$ and so on to refer to
  morphisms. The terminal object is written as $\emptycon$ with unique morphism
  $\epsilon : \Sub_\mbbc\,\Gamma\,\emptycon$. We omit the $\mbbc$ subscript if
  it is clear from context.

  In the initial model, an element of $\Con$ is a \emph{typing context}, an
  element of $\Sub$ is a \emph{parallel substitution}, $\emptycon$ is the
  \emph{empty typing context} and $\epsilon$ is the empty list viewed as a
  parallel substitution. The identity substitution $\id : \Sub\,\Gamma\,\Gamma$
  maps each variable to itself, and $\sigma \circ \delta$ can be computed
  by substituting each term in $\sigma$ (which is a list of terms) with $\delta$.
\item
  For each $i \in \{0,1\}$ and $j \in \mbb{N}$, we have $\Ty_{i,j} : \Con \to
  \Set$ and $\Tm_{i,j} : (\Gamma : \Con) \to \Ty_{i,j}\,\Gamma \to \Set$, as
  sets of types and terms. Both types and terms can be substituted:
  \begin{alignat*}{3}
    & \blank\![\blank\!] &&: \Ty_{i,j}\,\Delta \to \Sub\,\Gamma\,\Delta \to \Ty_{i,j}\,\Gamma \\
    & \blank\![\blank\!] &&: \Tm_{i,j}\,\Delta\,A \to (\sigma : \Sub\,\Gamma\,\Delta) \to \Tm_{i,j}\,\Gamma\,(A[\sigma])
  \end{alignat*}
  Additionally, we have $A[\id] = A$ and $A[\sigma \circ \delta] =
  A[\sigma][\delta]$, and we have the same equations for term substitution as
  well.

  We also have a \emph{comprehension structure}: for each $\Gamma : \Con$
  and $A : \Ty_{i,j}\,\Gamma$, we have the extended context $\Gamma \ext A :
  \Con$ such that there is a natural isomorphism $\Sub\,\Gamma\,(\Delta\ext A)
  \simeq (\sigma : \Sub\,\Gamma\,\Delta) \times
  \Tm_{i,j}\,\Gamma\,(A[\sigma])$. We will sometimes use
  $\blank\!\ext_0\!\blank$ and $\blank\!\ext_1\!\blank$ to disambiguate
  object-level and meta-level context extensions.
\item
  For each $j$ we have a \emph{lifting structure}, consisting of a natural
  transformation $\Lift : \Ty_{0,j}\,\Gamma \to \Ty_{1,j}\,\Gamma$, and an
  invertible natural transformation $\qut{\blank} : \Tm_{0,j}\,\Gamma\,A \to
  \Tm_{1,j}\,\Gamma\,(\Lift A)$, with inverse $\spl\blank$.
\end{itemize}
\end{definition}


The following notions are derivable:
\begin{itemize}
\item
  By moving left-to-right along $\Sub\,\Gamma\,(\Delta\ext A) \simeq (\sigma :
  \Sub\,\Gamma\,\Delta) \times \Tm_{i,j}\,\Gamma\,(A[\sigma])$,
  and starting from the identity morphism $\id : \Sub\,(\Gamma\ext A)\,(\Gamma\ext A)$, we recover
  the \emph{weakening substitution} $\p : \Sub\,(\Gamma\ext A)\,\Gamma$ and the \emph{zero variable}
  $\q : \Tm_{i,j}\,(\Gamma\ext A)\,(A[\p])$.
\item
  By weakening $\q$, we recover a notion of variables as De Bruijn
  indices. In general, the $n$-th De Bruijn index is defined as $\q[\p^{n}]$,
  where $\p^{n}$ denotes $n$-fold composition of $p$.
\item
  By moving right-to-left along $\Sub\,\Gamma\,(\Delta\ext A) \simeq (\sigma :
  \Sub\,\Gamma\,\Delta) \times \Tm_{i,j}\,\Gamma\,(A[\sigma])$, we recover the
  operation which extends a morphism with a term. In the initial model, this
  extends a parallel substitution with an extra term, thus justifying the view
  of substitutions as lists of terms. We denote the extension operation as
  $(\sigma,\,t)$ for $\sigma : \Sub\,\Gamma\,\Delta$ and $t : \Tm_{i,j}\,\Gamma\,(A[\sigma])$.
\end{itemize}

\begin{notation}
De Bruijn indices are rather hard to read, so we will sometimes use nameful notation
for binders and substitutions. For example, we may write $\Gamma \ext (x : A)
\ext (y : B)$ for a context, and subsequently write $B[x \mapsto t]$ for
substituting the $x$ variable for some term $t : \Tm_{i,j}\,\Gamma\,A$. Using
nameless notation, we would instead have $B : \Ty_{i,j}\,(\Gamma \ext A)$ and
$B[\id,\,t] : \Ty_{i,j}\,\Gamma$; here we recover single substitution by extending the identity
substitution $\id : \Sub\,\Gamma\,\Gamma$ with $t$.

We may also leave weakening implicit: if a type or term is in a context
$\Gamma$, we may use it in an extended context $\Gamma \ext A$ without marking
the weakening substitution.
\end{notation}

\begin{definition}
A \textbf{model of 2LTT} is a model of basic 2LTT which supports certain type
formers. For the sake of brevity, we specify here only a small collection of
type formers. However, we will argue later that our results extend to more
general notions of inductive types. We specify type formers in the following. We
omit substitution rules for type and term formers; all of these can be specified in
an evident structural way, e.g.\ as in $(\Lift A)[\sigma] = \Lift(A[\sigma])$.
\begin{itemize}
\item \emph{Universes}. For each $i$ and $j$, we have a Coquand-style universe
  \cite{coquand2018canonicity} in $\Ty_{i,j}$. This consists of $\U_{i,j} : \Ty_{i,j+1}\,\Gamma$,
  together with $\El : \Tm_{i,j+1}\,\Gamma\,\U_{i,j} \to \Ty_{i,j}\,\Gamma$ and $\Code$, where $\Code$
  and $\El$ are inverses.
\item \emph{$\Sigma$-types}. We have $\Sigma\,(x : A)\,B : \Ty_{i,j}\,\Gamma$ for $A : \Ty_{i,j}\,\Gamma$
  and $B : \Ty_{i,j}\,(\Gamma \ext (x : A))$, together with a natural isomorphism consisting of pairing and projections:
  \[ \Tm_{i,j}\,\Gamma\,(\Sigma\,(x : A)\,B) \simeq (t : \Tm_{i,j}\,\Gamma\,A) \times \Tm_{i,j}\,\Gamma\,(B[x\mapsto t]) \]
  We write $(t,\,u)$ for pairing and $\fst$ and $\snd$ for projections.
\item \emph{Function types}.
  We have $\Pi\,(x : A)\,B : \Ty_{i,j}\,\Gamma$ for $A : \Ty_{i,j}\,\Gamma$ and
  $B : \Ty_{i,j}\,(\Gamma \ext (x : A))$, together with abstraction and application
  \begin{alignat*}{3}
    & \lam &&: \Tm_{i,j}\,(\Gamma \ext (x : A))\,B \to \Tm_{i,j}\,\Gamma\,(\Pi\,(x : A)\,B) \\
    & \app &&: \Tm_{i,j}\,\Gamma\,(\Pi\,(x : A)\,B) \to \Tm_{i,j}\,(\Gamma \ext (x : A))\,B
  \end{alignat*}
  such that $\lam$ and $\app$ are inverses. Since we have explicit
  substitutions, this specification of $\app$ is equivalent to the
  ``traditional'' one:
  \[ \app' : (t : \Tm_{i,j}\,\Gamma\,(\Pi\,(x :
  A)\,B))\to (u : \Tm_{i,j}\,\Gamma\,A) \to \Tm_{i,j}\,\Gamma\,(B[x \mapsto
    u]).\]
  The traditional version is definable as $(\app\,t)[\id,\,u]$. We use
  the $\app$-$\lam$ isomorphism because it is formally more convenient.

\item \emph{Natural numbers}. We have $\Nat_{i,j} : \Ty_{i,j}\,\Gamma$, $\zero_{i,j} : \Tm_{i,j}\,\Gamma\,\Nat_{i,j}$,
  and $\suc_{i,j} : \Tm_{i,j}\,\Gamma\,\Nat_{i,j} \to \Tm_{i,j}\,\Gamma\,\Nat_{i,j}$. The eliminator is the following.
  \begin{alignat*}{5}
    & \NatElim :\hspace{0.5em}
                &&  (P &&: \Ty_{i,k}\,(\Gamma \ext (n : \Nat_{i,j}))) \\
    &           &&  (z &&: \Tm_{i,k}\,\Gamma\,(P[n\mapsto \zero_{i,j}])) \\
    &           &&  (s &&: \Tm_{i,k}\,(\Gamma \ext (n : \Nat_{i,j}) \ext (\mit{pn} : P[n\mapsto n]))\,(P[n\mapsto \suc_{i,j}\,n]))\\
    &           &&  (t &&: \Tm_{i,j}\,\Gamma\,\Nat_{i,j}) \\
    &\hspace{3em}\to        &&&& \Tm_{i,k}\,\Gamma\,(P[n\mapsto t]))
  \end{alignat*}
  We also have the $\beta$-rules:
  \begin{alignat*}{5}
    & \NatElim\,P\,z\,s\,\zero_{i,j}     &&= z \\
    & \NatElim\,P\,z\,s\,(\suc_{i,j}\,t) &&= s[n \mapsto t,\,\mit{pn} \mapsto \NatElim\,P\,z\,s\,t]
  \end{alignat*}
  Note that we can eliminate from a level $j$ to any level $k$.
\end{itemize}
\end{definition}

\begin{definition}
The \textbf{syntax of 2LTT} is the initial model of 2LTT, defined as a quotient
inductive-inductive type \cite{kaposi2019constructing}. The signature for this
type includes four sorts ($\Con$, $\Sub$, $\Ty$, $\Tm$), constructors for
contexts, explicit substitutions and type/term formers, and equality
constructors for all specified equations. The syntax comes with an
\emph{induction principle}, from which we can also derive the \emph{recursion principle} as a
non-dependent special case, witnessing the initiality of the syntax.
\end{definition}

\subsubsection{Comparison to Annekov et al.}
Comparing our models to the primary reference on 2LTT \cite{twolevel}, the main
difference is the handling of ``sizing'' levels. In ibid.\ there is a cumulative
lifting from $\Ty_{i,j}$ to $\Ty_{i,j+1}$, which we do not assume. Instead,
we allow elimination from $\Nat_{i,j}$ into any $k$ level. This means that we
can manually define lifting maps from $\Nat_{i,j}$ to $\Nat_{i,j+1}$ by
elimination. This is more similar to e.g.\ Agda, where we do not have cumulativity,
but we can define explicit lifting from $\Nat_j : \Set_j$ to $\Nat_k : \Set_k$.

In ibid., ``two-level type theory'' specifically refers to the setup
where the object level is a homotopy type theory and the meta level is an
extensional type theory. In contrast, we allow a wider range of setups under the
2LTT umbrella. Ibid.\ also considers a range of additional
strengthenings and extensions of 2LTT \cite[Section~2.4]{twolevel}, most of which
are useful in synthetic homotopy theory. We do not assume any of these and
stick to the most basic formulation of 2LTT.

\subsubsection{Elaborating Informal Syntax}
The justification for the usage of the informal syntax in Section
\ref{sec:tour-of-2ltt} is \emph{elaboration}. Elaboration translates surface
notation to the formal syntax (often termed ``core syntax'' in the
implementation of dependently typed programming languages). Elaboration is
partial: it may throw error on invalid input. Elaboration also requires
\emph{decidability of conversion} for the formal syntax, since type checking
requires comparing types for conversion, which in turn requires comparing terms.
We do not detail elaboration in this paper, we only make some observations.
\begin{itemize}
\item Decidability of conversion can be shown through normalization. Since 2LTT
      is a minor variation on ordinary MLTT, we expect that existing proofs for MLTT
      can be adapted without issue \cite{coquand2018canonicity,decidableconv}.
\item We need to elaborate the surface-level Russell-style universe notation to
      Coquand-style universes with explicit $\El$ and $\Code$ annotations. It is
      likely that this can be achieved with the use of bidirectional elaboration
      \cite{DBLP:journals/csur/DunfieldK21}, using two directions (checking and
      inference) both for types and terms, and inserting $\El$ and $\Code$ when
      switching between types and terms.
\item The formal complexity of elaboration varies wildly, depending on how
      explicit the surface syntax is. For example, Agda-style implicit arguments
      require metavariables and unification, which dramatically increases the
      difficulty of formalization.
\end{itemize}

\subsubsection{The Setoid Interpretation of the Syntax}\label{sec:2lttsetoid} We need to show
that the 2LTT syntax, as a particular inductive type, can be modeled in setoids.
Unfortunately, there is no prior work which presents the setoid interpretation
of quotient types in the generality of quotient induction-induction. We sketch
it for our specific use case, but we argue that the interpretation is fairly
mechanical and could be adapted to any quotient inductive-inductive type.
We focus on contexts and types here.

We mutually inductively define $\Con : \Set$, $\Con^{\sim} : \Con \to \Con \to \Prop$, $\Ty_{i,j} : \Con \to \Set$
      and $\Ty^{\sim} : \Ty_{i,j}\,\Gamma \to \Ty_{i,j}\,\Delta \to \Con^{\sim}\,\Gamma\,\Delta \to \Prop$, where
      $\Con^{\sim}$ and $\Ty^{\sim}$ are proof-irrelevant relations. Note that $\Ty^{\sim}$ is \emph{heterogeneous},
      since we can relate types in different (but provably convertible) contexts.

      We specify $\emptycon$ and $\blank\ext\blank$ in $\Con$. To $\Con^{\sim}$ we add that $\emptycon$ and $\blank\ext\blank$ are congruences, and also that $\Con^{\sim}$ is reflexive, symmetric and transitive.

We add all type formation rules to $\Ty_{i,j}$, and add all quotient equations to
$\Ty^{\sim}$ as homogeneously indexed rules. For example, the substitution rule
for $\Nat_0$ is specified as $ \Nat_0[] :
\Ty^{\sim}\,(\Nat_0[\sigma])\,\Nat_0\,\refl^{\sim} $, where $\sigma :
\Sub\,\Gamma\,\Delta$ and $\refl^{\sim} : \Con^{\sim}\,\Gamma\,\Gamma$. We also
add all congruences and reflexivity, symmetry and transitivity to $\Ty^\sim$.

We add a \emph{coercion} rule to types, as $\msf{coerce} :
\Con^{\sim}\,\Gamma\,\Delta \to \Ty_{i,j}\,\Gamma \to \Ty_{i,j}\,\Delta$. This is probably
familiar from the traditional specifications of type theories that use
conversion relations. However, we additionally need a \emph{coherence rule}
which witnesses $\Ty^{\sim}\,A\,(\msf{coerce}\,p\,A)\,p$ for each $A :
\Ty_{i,j}\,\Gamma$ and $p : \Con^{\sim}\,\Gamma\,\Delta$. In short, coherence
says that coercion preserves conversion.

The concise description of this setup is that $(\Con,\,\Con^{\sim})$ is a
setoid, and $(\Ty_{i,j},\,\Ty^{\sim})$ is a \emph{setoid fibration} over
$(\Con,\,\Con^{\sim})$, for each $i$ and $j$. This notion of setoid fibration is
equivalent to the ordinary notion of fibration, when setoids are viewed as
categories with proof-irrelevant invertible morphisms. This is appropriate,
since types in general are interpreted as setoid fibrations in the setoid model
of a type theory.

This scheme generalizes to substitutions and terms in a mechanical way. $\Sub$
is indexed over two contexts, so it is interpreted as a fibration over
the setoid product of $\Con$ and $\Con$. $\Tm$ is indexed over $(\Gamma :
\Con)\times \Ty_{i,j}\,\Gamma$, so it is interpreted as a fibration over a setoid
$\Sigma$-type. In every conversion relation we take the congruence closure of the specified
quotient equations. Finally, we can use the induction principle for the thus
specified sets and relations to interpret the internal elimination principle for
the 2LTT syntax.



\subsection{Models and Syntax of the Object Theory}\label{sec:objmod}

We also need to specify the object theory, which serves as the output language
of staging. In general, the object theory corresponding to a particular flavor
of 2LTT is simply the type theory that supports only the object-level $\Ty_{0,j}$
hierarchy with its type formers.

\begin{definition}
A \textbf{model of the object theory} is a category-with-families, with types and
terms indexed over $j \in \mbb{N}$, supporting Coquand-style universes $\U_j$,
type formers $\Pi$, $\Sigma$ and $\Nat_j$, with elimination from $\Nat_j$ to
any level $k$.
\end{definition}

\begin{definition}
Like before, the \textbf{syntax of the object theory} is the initial model,
given as a quotient inductive-inductive type, and it can be also interpreted
in setoids.
\end{definition}

\begin{notation}
From now on, by default we use $\Con$, $\Sub$, $\Ty$ and $\Tm$ to refer to sets
in the syntax of 2LTT. We use $\mbbo$ to refer to the object syntax, and
$\Con\ob$, $\Sub\ob$, $\Ty\ob$ and $\Tm\ob$ to refer to its underlying sets.
\end{notation}

\begin{definition}
We recursively define the \textbf{embedding} of object syntax into 2LTT syntax,
which preserves all structure strictly and consists of the following functions:
\begin{alignat*}{4}
  & \emb{\blank} : \Con\ob \to \Con  && \emb{\blank} : \Sub\ob\,\Gamma\,\Delta \to \Sub\,\emb{\Gamma}\,\emb{\Delta}\\
  & \emb{\blank} : \Ty_{\mbbo j}\,\Gamma \to \Ty_{0,j}\,\emb{\Gamma}\hspace{2em} && \emb{\blank} : \Tm_{\mbbo j}\,\Gamma\,A \to \Tm_{0,j}\,\emb{\Gamma}\,\emb{A}
\end{alignat*}
\end{definition}

\section{The Staging Algorithm}\label{sec:staging-algorithm}

In this section we specify what we mean by a staging algorithm, then proceed to
define one.

\begin{definition}\label{def:staging}
A \textbf{staging algorithm} consists of two functions:
  \begin{alignat*}{3}
    & \Stage : \Ty_{0,j}\,\emb{\Gamma} \to \Ty_{\mbbo j}\,\Gamma\hspace{2em} \Stage : \Tm_{0,j}\,\emb{\Gamma}\,A \to \Tm_{\mbbo j}\,\Gamma\,(\Stage\,A)
  \end{alignat*}
Note that we can stage open types and terms as long as their contexts are purely
object-level. By \emph{closed staging} we mean staging only for closed types and
terms.
\end{definition}
\begin{definition}
  The following properties are of interest when considering a staging algorithm:
  \begin{enumerate}
  \item \emph{Soundness:} $\emb{\Stage\,A} = A$ and $\emb{\Stage\,t} = t$.
  \item \emph{Stability:} $\Stage\,\emb{A} = A$ and $\Stage\,\emb{t} = t$.
  \item \emph{Strictness:} the function extracted from $\Stage$ preserves all
    type and term formers strictly (not just up to conversion).
  \end{enumerate}
\end{definition}

Soundness and stability together state that embedding is a bijection
on types and terms (up to conversion). This is a statement of
\emph{conservativity} of 2LTT over the object theory. In \cite{twolevel} a
significantly weaker conservativity theorem is shown, which only expresses that
there exists a function from $\Tm_{0,j}\,\emb{\Gamma}\,\emb{A}$ to $\Tm_{\mbbo
  j}\,\Gamma\,A$.

Strictness tells us that staging does not perform any computation in the object
theory; e.g.\ it does not perform $\beta$-reduction. This is practically quite
important; for instance, \emph{full normalization of 2LTT} is a sound and stable
staging algorithm, but it is not strict, and it is practically not useful, since
it does not provide any control over inlining and $\beta$-reduction in code
generation.

Note that staging necessarily preserves type and term formers up to conversion,
since it is defined on quotiented types and terms. More generally, staging
respects conversion because every function must respect conversion in the
meta type theory.

\subsection{The Presheaf Model}

We review the presheaf model described in \cite[Section~2.5.3]{twolevel}.  This
will directly yield a closed staging algorithm, by recursive evaluation of 2LTT types and
terms in the model.

We give a higher-level overview first. Presheaves can be viewed as a
generalization of sets, as ``variable sets'' which may vary over a base
category. A set-based model of 2LTT yields a standard interpreter, where every
syntactic object is mapped to the obvious semantic counterpart. The presheaf
model over the category of contexts and substitutions in $\mbbo$ generalizes
this to an interpreter where semantic values may depend on object-theoretic
contexts. This allows us to handle object-theoretic \emph{types and terms}
during interpretation, since these necessarily depend on their contexts.

In the presheaf model, every semantic construction must be \emph{stable} under
object-level substitution, or using alternative terminology, must be
\emph{natural} with respect to substitution. Semantic values can be viewed as
runtime objects during interpretation, and naturality means that
object-theoretic substitutions can be \emph{applied} to such runtime objects,
acting on embedded object-theoretic types and terms. Stability underlies
the core trade-off in staging:
\begin{itemize}
\item On one hand, every syntactic rule and construction in 2LTT must
      be stable, which restricts the range of metaprograms
      that can be written in 2LTT. For example, we cannot make decisions
      based on sizes of the object-theoretic scopes, since these sizes may be changed
      by substitution.
\item On the other hand, if only stable constructions are possible, we never
      have to prove stability, or even explicitly mention object-theoretic
      contexts, substitutions and variables. The tedious details of deep embeddings
      can be replaced with a handful of staging operations and typing rules.
\end{itemize}
In principle, the base category in a 2LTT can be any category with finite
products; it is not required that it is a category of contexts and substitutions.
If the base category has more structure, we get more interesting object
theories, but at the same time stability becomes more onerous, because we must
be stable under \emph{more} morphisms. If the base category is simpler, with
fewer morphisms, then the requirement of stability is less restrictive. In
Section \ref{sec:intensional-analysis}, we will look at a concrete example for this.

We proceed to summarize the key components of the model in a bit more detail. We
denote the components of the model by putting hats on 2LTT components, e.g.\ as
in $\hCon$.

\begin{notation}
In this section, we switch to naming elements of $\Cono$ as $a$, $b$ and $c$,
and elements of $\Subo$ as $f$, $g$, and $h$, to avoid name clashing with
contexts and substitutions in the presheaf model.
\end{notation}

\subsubsection{The Syntactic Category and the Meta-Level Fragment}

\begin{definition} $\bs{\wh{\Con} : \Set_{\omega+1}}$ is defined as the set of presheaves
over $\mbbo$. $\Gamma : \wh{\Con}$ has an action on objects $|\Gamma| :
\Cono \to \Set_\omega$ and an action on morphisms $\blank[\blank] : |\Gamma|\,b
\to \Subo\,a\,b \to |\Gamma|\,a$, such that $\gamma[\id] = \gamma$ and
$\gamma[f\circ g] = \gamma[f][g]$.

\begin{notation}
Above, we reused the substitution notation $\blank[\blank]$ for the action on
morphisms.  Also, we use lowercase $\gamma$ and $\delta$ to denote elements of
$|\Gamma|\,a$ and $|\Delta|\,a$ respectively.
\end{notation}

\end{definition}
\begin{definition} $\bs{\hSub\,\Gamma\,\Delta : \Set_\omega}$ is the set of natural transformations
from $\Gamma$ to $\Delta$. $\sigma : \hSub\,\Gamma\,\Delta$ has action
$|\sigma| : \{a : \Cono\} \to |\Gamma|\,a \to |\Delta|\,a$ such that
$|\sigma|\,(\gamma[f]) = (|\sigma|\,\gamma)[f]$.
\end{definition}

\begin{definition}
$\bs{\hTy_{1,j}\,\Gamma : \Set_\omega}$ is the set of \emph{displayed
presheaves} over $\Gamma$; see e.g.\ \cite[Chapter~1.2]{huber-thesis}. This is
equivalent to the set of presheaves over the category of elements of $\Gamma$,
but it is usually more convenient in calculations. An $A : \hTy\,\Gamma$ has
an action on objects $|A| : \{a : \Cono\} \to |\Gamma|\,a \to \Set_j$ and an
action on morphisms $\blank[\blank] : |A|\,\gamma \to (f : \Subo\,a\,b) \to
|A|\,(\gamma[f])$, such that $\alpha[\id] = \alpha$ and $\alpha[f \circ g] =
\alpha[f][g]$.

\begin{notation}
  We write $\alpha$ and $\beta$ respectively for elements of $|A|\,\gamma$ and
  $|B|\,\gamma$.
\end{notation}
\end{definition}

\begin{definition}
  $\bs{\hTm_{1,j}\,\Gamma\,A : \Set_\omega}$ is the set of sections of the
  displayed presheaf $A$. This can be viewed as a dependently typed analogue of
  a natural transformation.  A $t : \hTm_{1,j}\,\Gamma\,A$ has action $|t| :
  \{a\} \to (\gamma : |\Gamma|\,a) \to |A|\,\gamma$ such that $|t|\,(\gamma[f])
  = (|t|\,\gamma)[f]$.
\end{definition}

We also look at he empty context and context extension with meta-level types,
as these will appear in subsequent definitions.

\begin{definition}\label{def:psh-emptycon}
$\bs{{\widehat{\emptycon}} : \hCon}$ is defined as the presheaf
which is constantly $\top$, i.e.\ $|\wh{\emptycon}|\,\_ = \top$.
\end{definition}

\begin{definition}\label{def:psh-ext} For $A : \hTy_{1,j}\,\Gamma$, we define $\bs{\Gamma\,\,\wh{\ext}\,A}$
pointwise by $|\Gamma\,\,\wh{\ext}\,A|\,a := (\gamma : |\Gamma|\,a) \times |A|\,\gamma$
and $(\gamma,\,\alpha)[f] := (\gamma[f],\,\alpha[f])$.
\end{definition}

Using the above definitions, we can model the syntactic category of 2LTT, and
also the meta-level family structure and all meta-level type formers. For
expositions in previous literature, see \cite[Chapter~1.2]{huber-thesis} and
\cite[Section~4.1]{Hofmann97syntaxand}.

\subsubsection{The Object-Level Fragment}

We move on to modeling the object-level syntactic fragment of 2LTT. We make some
preliminary definitions. First, note that types in the object theory yield a
presheaf, and terms yield a displayed presheaf over them; this immediately
follows from the specification of a family structure in a cwf. Hence, we do a
bit of a name overloading, and have $\Ty_{\mbbo j} : \hCon$ and $\Tm_{\mbbo j} : \hTy\,\Ty_{\mbbo j}$.

\begin{definition}
$\bs{\hTy_{0,j}\,\Gamma : \Set_\omega}$ is defined as $\hSub\,\Gamma\,\Ty_{\mbbo j}$,
and $\bs{\hTm_{0,j}\,\Gamma\,A : \Set_\omega}$ is defined as $\hTm\,\Gamma\,(\Tm_{\mbbo j}[A])$.
\end{definition}

For illustration, if $A : \hTy_{0,j}\,\Gamma$, then $A :
\hSub\,\Gamma\,\Ty_{\mbbo j}$, so $|A| : \{a : \Cono\} \to |\Gamma|\,a \to
\Ty_{\mbbo j}\,a$. In other words, the action of $A$ on objects maps a semantic
context to a syntactic object-level type. Likewise, for $t :
\hTm_{0,j}\,\Gamma\,A$, we have $|t| : (\gamma : |\Gamma|\,a) \to \Tm_{\mbbo
  j}\,a\,(|A|\,\gamma)$, so we get a syntactic object-level term as output.

\begin{definition} For $A : \Ty_{0,j}\,\Gamma$, we define $\bs{\Gamma\,\,\wh{\ext}\,A}$
as $|\Gamma\,\,\wh{\ext}\,A|\,a := (\gamma : |\Gamma|\,a) \times \Tm_{\mbbo
  j}\,a\,(|A|\,\gamma)$ and $(\gamma,\,t)[f] := (\gamma[f],\,t[f])$. Thus,
extending a semantic context with an object-level type appends an object-level
term.
\end{definition}

Using the above definitions and following \cite{twolevel}, we can model all type
formers in $\hTy_{0,j}$. Intuitively, this is because $\hTy_{0,j}$ and
$\hTm_{0,j}$ return types and terms, so we can reuse the type and term formers
in the object theory.

\subsubsection{Lifting}

\begin{definition}\label{def:psh-lifting}
$\bs{\wh{\Lift} A : \hTy_{1,j}\,\Gamma}$ is defined as $\Tm_{\mbbo j}[A]$. With
  this, we get that $\hTm_{0,j}\,\Gamma\,A$ is equal to
  $\hTm_{1,j}\,\Gamma\,(\wh{\Lift}\,A)$. Hence, we can define both quoting and
  splicing as identity functions in the model.
\end{definition}

\subsection{Closed Staging}

\begin{definition}
The \textbf{evaluation morphism, denoted $\bs{\ev}$} is a morphism of 2LTT
models, between the syntax of 2LTT and $\hato$. It is defined by recursion on
the syntax (that is, appealing to the initiality of the syntax), and it strictly
preserves all structure.
\end{definition}

\begin{definition} \textbf{Closed staging} is defined as follows.
\begin{alignat*}{4}
  & \Stage : \Ty_{0,j}\,\emptycon \to \Ty_{\mbbo j}\,\emptycon \hspace{2em} && \Stage : \Tm_{0,j}\,\emptycon\,A \to \Tm_{\mbbo j}\,\emptycon\,(\Stage\,A) \\
  & \Stage\,A := |\ev\,A|\,\{\emptycon\}\,\tt && \Stage\,t := |\ev\,t|\,\{\emptycon\}\,\tt
\end{alignat*}
Note that $\wh{\emptycon}$ is defined as the terminal presheaf, which is
constantly $\top$. Also, $|\ev\,A|\{\emptycon\} : |\ev\,\emptycon|\,\emptycon \to \Ty_{\mbbo
  j}\,\emptycon$, therefore $|\ev\,A|\{\emptycon\} : \top \to \Ty_{\mbbo
  j}\,\emptycon$ and $|\ev\,t|\,\{\emptycon\} : \top \to \Tm_{\mbbo
  j}\,\emptycon\,(|\ev\,A|\,\tt)$.
\end{definition}

What about general (open) staging though? Given $A : \Ty_{0,j}\,\emb{\Gamma}$,
we get $|\ev\,A|\,\{\Gamma\} : |\emb{\Gamma}|\,\Gamma \to \Ty_{\mbbo j}$. We
need an element of $|\emb{\Gamma}|\,\Gamma$, in order to obtain object-level
type. Such ``generic'' semantic environments should be possible to construct, because
elements of $|\emb{\Gamma}|\,\Gamma$ are essentially lists of object-level terms
in $\Gamma$, by Definitions \ref{def:psh-ext} and \ref{def:psh-emptycon}, so $|\emb{\Gamma}|\,\Gamma$ should be
isomorphic to $\Subo\,\Gamma\,\Gamma$.  It turns out that this falls out from
the stability proof of $\ev$.

\subsection{Open Staging, Stability and Strictness}

We define a family of functions $\blank^P$ by induction on object syntax, such
that the interpretation of a context yields a generic semantic environment. The
induction motives are as follows.
\begin{alignat*}{4}
  & (\Gamma : \Cono)^P                    &&: |\ev\,\emb{\Gamma}|\,\Gamma && (\sigma : \Subo\,\Gamma\,\Delta)^P    &&: \Delta^P[\sigma] = |\ev\,\emb{\sigma}|\,\Gamma^P \\
  & (A      : \Ty_{\mbbo j}\,\Gamma)^P      &&: A = |\ev\,\emb{A}|\,\Gamma^P\hspace{2em} && (t      : \Tm_{\mbbo j}\,\Gamma\,A)^P   &&: t = |\ev\,\emb{t}|\,\Gamma^P
\end{alignat*}
We look at the interpretation of contexts.
\begin{itemize}
\item  For $\emptycon^P$, we need an element of $|\ev\,\emb{\emptycon}|\,\emptycon$, hence
  an element of $\top$, so we define $\emptycon^P$ as $\tt$.
\item
  For $(\Gamma \ext A)^P$, we need an element of
    \[(\gamma : |\ev{\emb{\Gamma}}|\,(\Gamma\ext A)) \times \Tm_{\mbbo j}\,(\Gamma \ext A)\,(|\ev\emb{A}|\,\gamma).\]
  We have $\Gamma^P : |\ev{\emb{\Gamma}}|\,\Gamma$, which we can weaken as
  $\Gamma^P[\p] : |\ev{\emb{\Gamma}}|\,(\Gamma \ext A)$, so we set the first
  projection of the result as $\Gamma^P[\p]$. For the second projection, the
  goal type can be simplified as follows:
  \begin{alignat*}{4}
    &       &&\Tm_{\mbbo j}\,(\Gamma \ext A)\,(|\ev\emb{A}|\,(\Gamma^P[\p])) && \\
    & =\,\, &&\Tm_{\mbbo j}\,(\Gamma \ext A)\,((|\ev\emb{A}|\,\Gamma^P)[\p]) &&\hspace{2em} \text{by naturality of $\ev\emb{A}$}\\
    & =\,\, &&\Tm_{\mbbo j}\,(\Gamma \ext A)\,(A[\p])                        &&\hspace{2em} \text{by $A^P$}
  \end{alignat*}
  We have the zero de Bruijn variable $\q : \Tm_{\mbbo j}\,(\Gamma \ext
  A)\,(A[\p])$.  Hence, we define $(\Gamma \ext A)^P$ as $(\Gamma^P[\p],\,\q)$.
\end{itemize}
Thus, a generic semantic context $\Gamma^P : |\ev\emb{\Gamma}|\,\Gamma$ is just
a list of variables, corresponding to the identity substitution $\id :
\Subo\,\Gamma\,\Gamma$ which maps each variable to itself.

The rest of the $\blank^P$ interpretation is straightforward and we omit it
here. In particular, preservation of definitional equalities is automatic, since
types, terms and substitutions are all interpreted as proof-irrelevant
equations.

\begin{definition}\label{def:open-staging} We define \textbf{open staging} as follows.
\begin{alignat*}{4}
  & \Stage : \Ty_{0,j}\,\emb{\Gamma} \to \Ty_{\mbbo j}\,\Gamma \hspace{2em} && \Stage : \Tm_{0,j}\,\emb{\Gamma}\,A \to \Tm_{\mbbo j}\,\Gamma\,(\Stage\,A) \\
  & \Stage\,A := |\ev\,A|\,\Gamma^P && \Stage\,t := |\ev\,t|\,\Gamma^P
\end{alignat*}
\end{definition}

\begin{theorem} Open staging is stable.
\end{theorem}
\begin{proof} For $A : \Ty_{\mbbo j}$, $\Stage\,\emb{A}$ is by definition $|\ev\,\emb{A}|\,\Gamma^P$,
  hence by $A^P$ it is equal to $A$. Likewise, $\Stage\,\emb{t}$ is equal to $t$ by $t^P$.
\end{proof}

\begin{theorem} Open staging is strict.
\end{theorem}
\begin{proof}
  This is evident from the definition of the presheaf model, where the action of
  each type and term former in $\wh{\Ty}_{0,j}$ and $\wh{\Tm}_{0,j}\,\Gamma\,A$ is
  defined using the corresponding type or term former in the object syntax.
\end{proof}

\subsection{Implementation and Efficiency}

We discuss now the operational behavior of the extracted staging algorithm
and look at potential adjustments and optimizations that make it more efficient
or more convenient to implement. Recall the types of the component functions in $\ev$:
\begin{alignat*}{3}
  & |\blank| \circ \ev : \Ty_{1 j}\,\Gamma \to \{\Delta : \Cono\} \to |\ev\,\Gamma|\,\Delta \to \Set_j\\
  & |\blank| \circ \ev : \Ty_{0 j}\,\Gamma \to \{\Delta : \Cono\} \to |\ev\,\Gamma|\,\Delta \to \Ty_{\mbbo j}\,\Delta\\
  & |\blank| \circ \ev : \Tm_{1 j}\,\Gamma\,A \to \{\Delta : \Cono\} \to (\gamma : |\ev\,\Gamma|\,\Delta) \to
    |\ev\,A|\,\gamma \\
  & |\blank| \circ \ev : \Tm_{0 j}\,\Gamma\,A \to \{\Delta : \Cono\} \to (\gamma : |\ev\,\Gamma|\,\Delta) \to
    \Tm_{\mbbo\,j}\,\Delta\,(|\ev\,A|\,\gamma)
\end{alignat*}
We interpret syntactic types or terms in a semantic environment
$\gamma : |\ev\,\Gamma|\,\Delta$. These environments are lists containing
a mix of  object-level terms and semantic values. The semantic values are represented
using metatheoretic inductive types and function types.
\begin{itemize}
  \item $|\Nat_1|\,\_$ is simply the set of natural numbers $\mbb{N}$.
  \item $|\Sigma_1\,A\,B|\,\gamma$ is simply a set of pairs of values.
  \item A non-dependent function type $A\,\wh{\to}\,B$ is defined as the presheaf
        exponential. The computational part of $|A\,\wh{\to}\,B|\,\{\Delta\}\,\gamma$
        is given by a function with type
        \[ (\Theta : \Cono) \to (\sigma : \Subo\,\Theta\,\Delta) \to |A|\,(\gamma[\sigma]) \to |B|\,(\gamma[\sigma])        \]
        This may be also familiar as the semantic implication from the Kripke
        semantics of intuitionistic logics. Whenever we evaluate a function
        application, we supply an extra $\id : \Subo\,\Delta\,\Delta$.
        This may incur cost via the $\gamma[\id]$ restrictions in a naive
        implementation, but this is easy to optimize, by introducing a formal
        representation of $\id$, such that $\gamma[\id]$ immediately computes to
        $\gamma$. The case of dependent functions is analogous operationally.
\end{itemize}
In summary, in the meta-level fragment of 2LTT, $\ev$ yields a reasonably efficient
computation of closed values, which reuses functions from the ambient
metatheory. Alternatively, instead of using ambient functions, we could use our
own implementation of function closures during staging.

In contrast, we have a bit of an efficiency problem in the handling of
object-level binders: whenever we go under such binder we have to weaken the
current semantic environment. Concretely, when moving from $\Delta$ to
$\Delta \ext A$, we have to shift $\gamma : |\ev\,\Gamma|\,\Delta$ to
$\gamma[\p] : |\ev\,\Gamma|\,(\Delta \ext A)$. This cannot be avoided by an easy
optimization trick. For object-level entries in the environment, this is cheap,
because we just add an extra explicit weakening, but for semantic values
weakening may need to perform deep traversals.

The same efficiency issue arises in formal presheaf-based
normalization-by-evaluation, such as in \cite{kaposinbe} and
\cite{coquand2018canonicity}. However, in practical implementations this issue
can be fully solved by using De Bruijn \emph{levels} in the semantic domain,
thus arranging that weakening has no operational cost on semantic values; see
\cite{coquand1996algorithm}, \cite{modulartc} and \cite{untypedtc} for
implementations in this style. We can use the same solution in staging. It is a
somewhat unfortunate mismatch that indices are far more convenient in
formalization, but levels are more efficient in practice.

In our prototype implementation, we use the above optimization with De Bruijn
levels, and we drop explicit substitutions from the 2LTT core syntax. In the
object-level syntax, we use De Bruijn levels for variables and closures in
binders. This lets us completely skip the weakening of environments, which is
slightly more efficient than the more faithfully extracted algorithm.

Additionally, we use an untyped tagged representation of semantic
values, since in Haskell (the implementation language) we do not have
strong enough types to represent typed semantic values. For example, there
are distinct data constructors storing semantic function values, natural number
literals and quoted expressions.

\paragraph{Caching} In a production-strength implementation we would need some
caching mechanism, to avoid excessive duplication of code. For example, if we
use the staged $\mit{map}$ function multiple times with the same type and
function arguments, we do not want to generate separate code for each
usage. De-duplicating object-level code with function types is usually safe,
since function bodies become closed top-level code after closure conversion. We
leave this to future work.

\section{Soundness of Staging}\label{sec:soundness}

In this section we prove soundness of staging. We build a proof-relevant logical
relation between the evaluation morphism $\ev$ and a \emph{restriction}
morphism, which restricts 2LTT syntax to object-level contexts. The relational
interpretations of $\Ty_{0,j}$ and $\Tm_{0,j}$ will yield the soundness
property.

\subsection{Working in $\hat{\mbbo}$}

We have seen that 2LTT can be modeled in the presheaf topos
$\hato$. Additionally, $\hato$ supports all type formers of extensional type
theory and certain other structures which are stable under object-level
substitution. As all constructions in this section must be stable, it makes
sense to work internally to $\hato$. This style has been previously used in
normalization proofs \cite{coquand2018canonicity} and also in the metatheory of
cubical type theories
\cite{licata2018internal,orton_et_al:LIPIcs:2016:6564,DBLP:conf/csl/CavalloMS20}.

When we work internally in a model of a type theory, we do not explicitly refer
to contexts, types, and substitutions. For example, when working in Agda, we do
not refer to Agda's typing contexts. Instead, we only work with terms, and use
functions and universes to abstract over types and semantic contexts. Hence, we
have to convert along certain isomorphisms when we switch between the internal
and external views. In the following, we summarize features in $\hato$ and also
the internal-external conversions.

We write $\whset_j$ for ordinal-indexed Russell-style universes. Formally, we
have Coquand-style universes, but for the sake of brevity we omit $\El$ and
$\Code$ from internal syntax. Universes are cumulative, and closed under $\Pi$,
$\Sigma$, extensional identity $\blank\!=\!\blank$ and inductive types. We use
the same conventions as in Notation \ref{basic-notation}.

\noindent The basic scheme for internalization is as follows:
\begin{alignat*}{3}
  & \Gamma : \hCon                  && \hspace{1em}\text{is internalized as}\hspace{1em} \Gamma : \whset_\omega\\
  & \sigma : \hSub\,\Gamma\,\Delta  && \hspace{1em}\text{is internalized as}\hspace{1em} \sigma : \Gamma \to \Delta\\
  & A : \hTy_{1,j}\,\Gamma           && \hspace{1em}\text{is internalized as}\hspace{1em} A      : \Gamma \to \whset_j \\
  & t : \hTm_{1,j}\,\Gamma\,A        && \hspace{1em}\text{is internalized as}\hspace{1em} t      : (\gamma : \Gamma) \to A\,\gamma
\end{alignat*}
\subsubsection{Object-Theoretic Syntax} The syntax of the object theory is clearly
fully stable under object-theoretic substitution, so we can internalize all of
its type and term formers. We internalize object-level types and terms as
$\Ty_{\mbbo j} : \whset_0$ and $\Tm_{\mbbo j} : \Ty_{\mbbo j} \to
\whset_0$. $\Ty_{\mbbo j}$ is closed under type formers. For instance, we have
\[ \Nat_j : \Ty_{\mbbo j}\hspace{1.5em}\zero_j : \Tm_{\mbbo j}\,\Nat_j\hspace{1.5em}
    \suc_j : \Tm_{\mbbo j}\,\Nat_j \to \Tm_{\mbbo j}\,\Nat_j \]
together with $\NatElim$, and likewise we have all other type formers.
\subsubsection{Internal $\ev$} We also use an internal view of
$\ev$ which maps 2LTT syntax to internal values; i.e.\ we compose $\ev$ with
internalization. Below, note that the input to $\ev$ is \emph{external}, so we
mark $\ev$ as being parameterized by external input.
\begin{alignat*}{3}
  &\ev\,(\Gamma &&: \Con)                      &&: \whset_\omega \\
  &\ev\,(\sigma &&: \Sub\,\Gamma\,\Delta)      &&: \ev\,\Gamma \to \ev\,\Delta \\
  &\ev\,(A      &&: \Ty_{1,j}\,\Gamma)          &&: \ev\,\Gamma \to \whset_j\\
  &\ev\,(t      &&: \Tm_{1,j}\,\Gamma)          &&: (\gamma : \ev\,\Gamma) \to \ev\,A\,\gamma\\
  &\ev\,(A      &&: \Ty_{0,j}\,\Gamma)          &&: \ev\,\Gamma \to \Ty_{\mbbo_j}\\
  &\ev\,(t      &&: \Tm_{0,j}\,\Gamma\,A)       &&: (\gamma : \ev\,\Gamma) \to \Tm_{\mbbo j}\,(\ev\,A\,\gamma)
\end{alignat*}
\subsubsection{Object-Level Fragment of 2LTT}
The purely object-level syntactic fragment of 2LTT can be internalized as
follows. We define externally the presheaf of object-level types as
$|\Ty_{0,j}|\,\Gamma := \Ty_{0,j}\,\emb{\Gamma}$, and the displayed presheaf of
object-level terms over $\Ty_{0,j}$ as $|\Tm_{0,j}|\,\{\Gamma\}\,A :=
\Tm_{0,j}\,\emb{\Gamma}\,A$. Hence, internally we have $\Ty_{0,j} : \whset_0$ and
$\Tm_{0,j} : \Ty_{0,j} \to \whset_0$. $\Ty_{0,j}$ is closed under all type formers,
analogously as we have seen for $\Ty_{\mbbo j}$.



\subsubsection{Embedding}
Now, the embedding operation $\emb{\blank}$ can be also internalized on types
and terms, as $\emb{\blank} : \Ty_{\mbbo j} \to \Ty_{0,j}$, and $\emb{\blank}
:\Tm_{\mbbo\,j}\,A \to \Tm_{0,j}\,\emb{A}$. Embedding strictly preserves all structure.

\subsection{The Restriction Morphism}

We define a family of functions $\re$ from the 2LTT syntax to objects in
$\hato$. We will relate this to the evaluation morphism $\ev$ in the relational
interpretation. In short, $\re$ restricts 2LTT syntax so that it can only depend
on object-level contexts, which are given as $\emb{\Gamma}$.

\begin{definition}\label{def:restriction} We specify the types of the \textbf{restriction operations} internally, and the $|\blank|$
components of the operations externally. The naturality of $|\blank|$ is
straightforward in each case.
\begingroup
\begin{alignat*}{5}
  & \re\,(\Gamma : \Con) : \whset_0 && \re\,(\sigma : \Sub\,\Gamma\,\Delta) : \re\,\Gamma \to \re\,\Delta \\
  & |\re\,\Gamma|\,\Delta := \Sub\,\emb{\Delta}\,\Gamma && |\re\,\sigma|\,\gamma := \sigma \circ \gamma \\
  & \\
  & \re\,(A : \Ty_{1,j}\,\Gamma) : \re\,\gamma \to \whset_0 && \re\,(t : \Tm_{1,j}\,\Gamma\,A) : (\gamma : \re\,\Gamma) \to \re\,A\,\gamma \\
  & |\re\,A|\,\{\Delta\}\,\gamma := \Tm_{1,j}\,\emb{\Delta}\,(A[\gamma])\hspace{2em} && |\re\,t|\,\gamma := t[\gamma]\\
  & \\
  & \re\,(A : \Ty_{0,j}\,\Gamma) : \re\,\Gamma \to \Ty_{0,j}\hspace{2em} && \re\,(t : \Tm_{0,j}\,\Gamma\,A) : (\gamma : \re\,\Gamma) \to \Tm_{0,j}\,(\re\,A\,\gamma)\\
  & |\re\,A|\,\gamma := A[\gamma] && |\re\,t|\,\gamma := t[\gamma]
\end{alignat*}
\endgroup
\end{definition}

\subsubsection{Preservation Properties of $\re$}
First, we note that $\re$ strictly preserves $\id$, $\blank\circ\blank$ and
type/term substitution, and it preserves $\emptycon$ and $\blank\ext_1\blank$ up
to isomorphism. We have the following isomorphisms internally to $\hato$:
\begin{alignat*}{3}
  & \re_{\emptycon} &&: \re\,\emptycon \simeq \top \\
  & \re_{\ext_1}   &&: \re\,(\Gamma \ext_1 A) \simeq ((\gamma : \re\,\Gamma) \times \re\,A\,\gamma)
\end{alignat*}

\begin{notation} When we have an isomorphism $f : A \simeq B$, we may write $f$ for the
  function in $A \to B$, and $f^{-1} : B \to A$ for its inverse.
\end{notation}

\begin{notation}
  We can use a pattern matching notation on isomorphisms. For example, if $f : A
  \simeq B$, then we may write $(\lambda\,(f\,a).\,t) : B \to C$, and likewise
  $(\lambda\,(f^{-1}\,b).\,t) : A \to C$, where the function bodies can refer to
  the bound $a : A$ and $b : B$ variables.
\end{notation}

The preservation properties $\re_{\emptycon}$ and $\re_{\ext_1}$ mean that $\re$
is a \emph{pseudomorphism} in the sense of \cite{gluing}, between the syntactic
cwf given by $(\Ty_{1,j},\,\Tm_{1,j})$ and the corresponding cwf structure in
$\hato$. In ibid.\ there is an analysis of such cwf morphisms, from which we
obtain the following additional preservation properties:
\begin{itemize}
\item Meta-level $\Sigma$-types are preserved up to isomorphism, so we have
  \[ \re_{\Sigma} : \re\,(\Sigma\,A\,B)\,\gamma
  \simeq ((\alpha : \re\,A) \times \re\,B\,(\rexti\,(\gamma,\,\alpha))). \]
  The semantic values of $\re\,(\Sigma\,A\,B)\,\gamma$
  are 2LTT terms with type $(\Sigma\,A\,B)[\gamma]$, restricted to object-level contexts.
  We can still perform pairing and projection with such restricted terms; hence
  the preservation property.
\item Meta-level $\Pi$-types and universes are preserved in a lax way. For $\Pi$, we have
  \[ \re_{\app} : \re\,(\Pi\,A\,B)\,\gamma \to (\alpha : \re\,A\,\gamma) \to \re\,B\,(\rexti\,(\gamma,\,\alpha))
  \]
  such that $\re_{\app}\,(\re\,t\,\gamma)\,\alpha =
  \re\,(\app\,t)\,(\rexti\,(\gamma,\,\alpha))$. In this case, we can apply a
  restricted term with a $\Pi$-type to a restricted term, but we cannot do
  lambda-abstraction, because that would require extending the context with a
  meta-level binder. For $\U_{1,j}$, we have
  \[
    \re_{\El} : \re\,\U_{1,j}\,\gamma \to \whset_j
  \]
  such that $\re_{\El}\,(\re\,t\,\gamma) = \re\,(\El\,t)\,\gamma$. Here, we only
  have lax preservation simply because $\whset_j$ is much larger than the the
  set of syntactic 2LTT types, so not every semantic $\whset_j$ has a syntactic
  representation.

\item Meta-level positive (inductive) types are preserved in an oplax way. In the case of
  natural numbers, we have
  \[ \re_{\mbb{N}} : \mbb{N} \to \re\,\Nat_{1,j}\,\gamma. \]
  This is a ``serialization'' map: from a metatheoretic natural number we
  compute a numeral as a closed 2LTT term. This works analogously for other
  inductive types, including parameterized ones. For an example, from a semantic
  list of restricted terms we would get a syntactic term with list type,
  containing the same restricted terms.
\end{itemize}
\subsubsection{Action on Lifting} We have an isomorphism $\re\,(\Lift\,A)\,\gamma \simeq
\Tm_{0,j}\,(\re\,A\,\gamma)$. This is given by quoting and splicing: we convert
between restricted meta-level terms with type $\Lift\,A$ and restricted
object-level terms with type $A$. Hence, we write components of this isomorphism
as $\spl\blank$ and $\qut{\blank}$, as internal analogues of the external
operations. With this, we also have $\re\,\qut{t}\,\gamma =
\qut{\re\,t\,\gamma}$ and $\re\,\spl\,t\,\gamma = \spl(\re\,t\,\gamma)$.

\subsubsection{Action on $\blank\ext_0\blank$} We have preservation up to isomorphism:
\[
  \re_{\ext_0} : \re\,(\Gamma \ext_0 A) \simeq ((\gamma : \re\,\Gamma)\times \Tm_{0,j}\,(\re\,A\,\gamma))
\]
This is because substitutions targeting $\Gamma \ext_0 A$ are the same as
pairs of substitutions and terms, by the specification of $\blank\!\ext_0\!\blank$.

\subsubsection{Action on Object-Level Types and Terms} $\re$ preserves all structure
in the object-level fragment of 2LTT. This follows from the $\re$ specification:
an external type $A : \Ty_{0,j}\,\Gamma$ is directly internalized as an element
of $\Ty_{0,j}$, and the same happens for $t : \Tm_{0,j}\,\Gamma\,A$.

\subsection{The Logical Relation}

Internally to $\hato$, we define by induction on the syntax of 2LTT a
proof-relevant logical relation interpretation, written as $\blank^{\approx}$. The induction
motives are specified as follows.
\begin{alignat*}{5}
  &(\Gamma &&: \Con)\rel && : \ev\,\Gamma \to \re\,\Gamma \to \whset_\omega\\
  &(\sigma &&: \Sub\,\Gamma\,\Delta)\rel &&: \Gamma\rel\,\gamma\,\gamma' \to \Delta\rel\,(\ev\,\sigma\,\gamma)\,(\re\,\sigma\,\gamma')\\
  &(A      &&: \Ty_{1,j}\,\Gamma)\rel &&: \Gamma\rel\,\gamma\,\gamma' \to \ev\,A\,\gamma \to \re\,A\,\gamma' \to \whset_j\\
  &(t      &&: \Tm_{1,j}\,\Gamma\,A)\rel &&: (\gamma\rel : \Gamma\rel\,\gamma\,\gamma') \to A\rel\,\gamma\rel\,(\ev\,t\,\gamma)\,(\re\,t\,\gamma')\\
  &(A      &&: \Ty_{0,j}\,\Gamma)\rel &&: \Gamma\rel\,\gamma\,\gamma' \to \emb{\ev\,A\,\gamma} = \re\,A\,\gamma' \\
  &(t      &&: \Tm_{0,j}\,\Gamma\,A)\rel &&: \Gamma\rel\,\gamma\,\gamma' \to \emb{\ev\,t\,\gamma} = \re\,t\,\gamma'
\end{alignat*}
For $\Con$, $\Sub$ and meta-level types and terms, this is a fairly standard
logical relation interpretation: contexts are mapped to relations, types to
dependent relations, and substitutions and terms respect relations. We will only have
modest complications in meta-level type formers because we will sometimes need to use
the lax/oplax preservations properties of $\re$. For object-level types and
terms, we get soundness statements: evaluation via $\ev$ followed by embedding
back to 2LTT is the same as restriction to object-level contexts. We describe
the $\blank\rel$ interpretation in the following.

\subsubsection{Syntactic Category and Terminal Object}
Here, we simply have $\id\rel\,\gamma\rel := \gamma\rel$ and $(\sigma \circ
\delta)\rel\,\gamma\rel := \sigma\rel\,(\delta\rel\,\gamma\rel)$. The terminal
object is interpreted as $\emptycon\rel\,\gamma\,\gamma' := \top$.

\subsubsection{Meta-Level Family Structure}
We interpret context extension and type/term substitution as follows. Note the
usage of the pattern matching notation on the $\rexti$ isomorphism.
\begin{alignat*}{4}
  & (\Gamma \ext_1 A)\rel\,(\gamma,\,\alpha)\,(\rexti(\gamma',\,\alpha')) && :=
    (\gamma\rel : \Gamma\rel\,\gamma\,\gamma') \times A\rel\,\gamma\rel\,\alpha\,\alpha'\\
  & (A[\sigma])\rel\,\gamma\rel\,\alpha\,\alpha' && := A\rel\,(\sigma\rel\,\gamma\rel)\,\alpha\,\alpha'\\
  & (t[\sigma])\rel\,\gamma\rel\                 && := t\rel\,(\sigma\rel\,\gamma\rel)
\end{alignat*}
It is enough to specify the $\blank\rel$ action on extended substitutions
$(\sigma,\,t) : \Sub\,\Gamma\,(\Delta\ext A)$, $\p : \Sub\,(\Gamma\ext
A)\,\Gamma$, $\p : \Sub\,(\Gamma\ext A)\,\Gamma$ and $\q : \Tm_{1,j}\,(\Gamma
\ext A)\,(A[\p])$. The category-with-families equations hold evidently.
\begin{alignat*}{4}
  & (\sigma,\,t)\rel\,\gamma\rel := (\sigma\rel\,\gamma\rel,\,t\rel\,\gamma\rel) \hspace{2em} \p\rel\,(\gamma\rel,\,t\rel) := \gamma\rel \hspace{2em} \q\rel\,(\gamma\rel,\,t\rel) := t\rel
\end{alignat*}

\subsubsection{Meta-Level Natural Numbers} First, we have to define a relation:
\begin{alignat*}{4}
  & (\Nat_{1,j})\rel : \Gamma\rel\,\gamma\,\gamma' \to \mbb{N} \to \re\,\Nat_{1,j}\,\gamma' \to \whset_j\\
  & (\Nat_{1,j})\rel\,\gamma\rel\,n\,n' := (\re_{\mbb{N}}\,n = n')
\end{alignat*}
Note that evaluation sends $\Nat_{1,j}$ to the semantic type of natural numbers,
i.e.\ $\ev\,\Nat_{1,j} = \mbb{N}$. We refer to the serialization map
$\re_{\mbb{N}} : \mbb{N} \to \re\,\Nat_{1,j}\,\gamma'$ above. In short,
$(\Nat_{1,j})\rel$ expresses \emph{canonicity}: $n'$ is canonical precisely if
it is of the form $\re_{\mbb{N}}\,n$ for some $n$.

For $\zero\rel$ and $\suc\rel$, we need to show that serialization preserves
$\zero$ and $\suc$ respectively, which is evident.

Let us look at elimination. We need to define the following:
\begin{alignat*}{4}
  &(\msf{NatElim}\,P\,z\,s\,n)\rel\,\gamma\rel :
  P\rel\,(\gamma\rel,\,n\rel\,\gamma\rel)\,(\ev\,(\msf{NatElim}\,P\,z\,s\,n)\,\gamma)\,
                                           (\re\,(\msf{NatElim}\,P\,z\,s\,n)\,\gamma')
\end{alignat*}
Unfolding $\ev$, we can further compute this to the following:
\begin{alignat*}{4}
  &(\msf{NatElim}\,P\,z\,s\,n)\rel\,\gamma\rel : \\
  &\hspace{2em}P\rel\,(\gamma\rel,\,n\rel\,\gamma\rel)\,
  (\msf{NatElim}\,(\lambda\,n.\,\ev\,P\,(\gamma,\,n))\,
                  (\ev\,z\,\gamma)\,
                  (\lambda\,n\,\mit{pn}.\,\ev\,s\,((\gamma,\,n),\,\mit{pn}))
                  (\ev\,n\,\gamma))\\
  &\hspace{8em}(\re\,(\msf{NatElim}\,P\,z\,s\,n)\,\gamma')
\end{alignat*}
In short, we need to show that $\msf{NatElim}$ preserves relations. Here we
switch to the external view temporarily. By Definition \ref{def:restriction}, we know that
\[
   |\re\,(\msf{NatElim}\,P\,z\,s\,n)|\,\gamma' = (\msf{NatElim}\,P\,z\,s\,n))[\gamma'].
\]
At the same time, we have $n\rel\,\gamma\rel : \re_{\mbb{N}}\,(\ev\,n\,\gamma) =
\re\,n\,\gamma'$, hence we know externally that $|\ev\,n|\,\gamma = n[\gamma']$.
In other words, $n[\gamma']$ is canonical and is obtained as the serialization
of $\ev\,n\,\gamma$. Therefore, $(\msf{NatElim}\,P\,z\,s\,n)[\gamma']$ is
definitionally equal to $|\ev\,n|\,\gamma$-many applications of $s$ to
$z$, and we can use $|\ev\,n|\,\gamma$-many applications of $s\rel$ to $z\rel$
to witness the goal type. The $\beta$-rules for $\msf{NatElim}$ are also respected
by this definition.

\subsubsection{Meta-Level $\Sigma$-Types} We define relatedness pointwise. Pairing and projection
are interpreted as meta-level pairing and projection.
\begin{alignat*}{3}
  &(\Sigma\,A\,B)\rel\,\gamma\rel : ((\alpha : \ev\,A\,\gamma) \times \ev\,B\,(\gamma,\,\alpha))
             \to \re\,(\Sigma\,A\,B)\,\gamma' \to \whset_j\\
  &(\Sigma\,A\,B)\rel\,\gamma\rel\,(\alpha,\,\beta)\,(\re_{\Sigma}^{-1}(\alpha',\,\beta')) :=
             (\alpha\rel : A\rel\,\gamma\rel\,\alpha\,\alpha') \times B\rel\,(\gamma\rel,\,\alpha\rel)\,\beta\,\beta'
\end{alignat*}

\subsubsection{Meta-Level $\Pi$-Types} We again use a pointwise definition. Note that we need to use
$\re_{\app}$ to apply $t' : \re\,(\Pi\,A\,B)\,\gamma'$ to $\alpha'$.
\begin{alignat*}{3}
  &(\Pi\,A\,B)\rel\,\gamma\rel : ((\alpha : \ev\,A\,\gamma) \to \ev\,B\,(\gamma,\,\alpha))
             \to \re\,(\Pi\,A\,B)\,\gamma' \to \whset_j\\
  &(\Pi\,A\,B)\rel\,\gamma\rel\,t\,t' :=
    (\alpha : \ev\,A\,\gamma)(\alpha' : \re\,A\,\gamma')(\alpha\rel : A\rel\,\gamma\rel\,\alpha\,\alpha') \to
    B\rel\,(\gamma\rel,\,\alpha\rel)\,(t\,\alpha)\,(\re_{\app}\,t'\,\alpha')
\end{alignat*}
For abstraction and application, we use a curry-uncurry definition:
\begin{alignat*}{4}
  &(\lam\,t)\rel\,\gamma\rel                &&:= \lambda\,\alpha\,\alpha'\,\alpha\rel.\,t\rel\,(\gamma\rel,\,\alpha\rel)\\
  &(\app\,t)\rel\,(\gamma\rel,\,\alpha\rel) &&:= t\rel\,\gamma\rel\,\alpha\,\alpha'\,\alpha\rel
\end{alignat*}

\subsubsection{Meta-Level Universes} We interpret $\U_{1,j}$ as a semantic relation space:
\begin{alignat*}{3}
  &(\U_{1,j})\rel\,\gamma\rel : \whset_j \to \re\,\U_{1,j}\,\gamma' \to \whset_{j+1} \\
  &(\U_{1,j})\rel\,\gamma\rel\,t\,t' := t \to \re_{\El}\,t' \to \whset_{j}
\end{alignat*}
Note that we have
\begin{alignat*}{3}
  & (\El\,t)\rel   &&: (\gamma\rel : \Gamma\rel\,\gamma\,\gamma') \to \ev\,t\,\gamma \to \re\,(\El\,t)\,\gamma' \to \whset_j\\
  & (\Code\,t)\rel &&: (\gamma\rel : \Gamma\rel\,\gamma\,\gamma') \to \ev\,t\,\gamma \to \re\,(\El\,t)\,\gamma' \to \whset_j.
\end{alignat*}
The types coincide because of the equation $\re\,(\El\,t)\,\gamma' =
\re_{\El}\,(\re\,t\,\gamma')$. Therefore we can interpret $\El$ and $\Code$ as
identity maps, as $(\El\,t)\rel := t\rel$ and $(\Code\,t)\rel := t\rel$.

\subsubsection{Object-Level Family Structure} We interpret extended contexts as follows.
\begin{alignat*}{3}
  & (\Gamma \ext_0 A)\rel\,(\gamma,\,\alpha)\,(\rextizero\,(\gamma',\,\alpha')) := (\gamma\rel : \Gamma\rel\,\gamma\,\gamma') \times (\emb{\alpha} = \alpha')
\end{alignat*}
Note that $\alpha : \Tm_{\mbbo j}\,(\ev\,A\,\gamma)$, so $\emb{\alpha} :
\Tm_{0,j}\,\emb{\ev\,A\,\gamma}$, but since $A\rel\,\gamma\rel :
\emb{\ev\,A\,\gamma} = \re\,A\,\gamma'$, we also have $\emb{\alpha} :
\Tm_{0,j}\,(\re\,A\,\gamma')$. Thus, the equation $\emb{\alpha} = \alpha'$ is
well-typed. For type substitution, we need
\[
  (A[\sigma])\rel\,\gamma\rel : \emb{\ev\,(A[\sigma])\,\gamma} = \re\,(A[\sigma])\,\gamma'.
\]
The goal type computes to $\emb{\ev\,A\,(\ev\,\sigma\,\gamma)} =
\re\,A\,(\re\,\sigma\,\gamma')$. This is obtained directly from
$A\rel\,(\sigma\rel\,\gamma\rel)$.  Similarly,
\begin{alignat*}{6}
  & (t[\sigma])\rel\,\gamma\rel &&:= t\rel\,(\sigma\rel\,\gamma\rel) && \p\rel\,(\gamma\rel,\,\alpha\rel) &&:= \gamma\rel \\
  & (\sigma,\,t)\rel\,\gamma\rel &&:= (\sigma\rel\,\gamma\rel,\,t\rel\,\gamma\rel)\hspace{2em}&& \q\rel\,(\gamma\rel,\,\alpha\rel) &&:= \alpha\rel
\end{alignat*}

\subsubsection{Lifting Structure}
\begin{alignat*}{3}
  &(\Lift\,A)\rel : \Gamma\rel\,\gamma\,\gamma' \to \Tm_{\mbbo j}\,(\ev\,A\,\gamma) \to \re\,(\Lift\,A)\,\gamma' \to \whset_j\\
  &(\Lift\,A)\rel\,\gamma\rel\,t\,t' := (\emb{t} = \spl t')
\end{alignat*}
This is well-typed by $A\rel\,\gamma\rel : \emb{\ev\,A\,\gamma} = \re\,A\,\gamma'$, which implies
that $\emb{t} : \Tm_{\mbbo j}\,(\re\,A\,\gamma')$. For $\qut{t}$, we need
\[  \qut{t}\rel\,\gamma\rel : \emb{\ev\,\qut{t}\,\gamma} = \spl(\re\,\qut{t}\,\gamma'). \]
The goal type can be further computed to $\emb{\ev\,t\,\gamma} = \re\,t\,\gamma'$, which we prove
by $t\rel\,\gamma\rel$. For splicing, we need
\[
  (\spl t)\rel\,\gamma\rel : \emb{\ev\,\spl t\,\gamma} = \re\,\spl t\,\gamma'
\]
where the goal type computes to $\emb{F\,t\,\gamma} = \spl(\re\,t\,\gamma')$,
but this again follows directly from $t\rel\,\gamma\rel$.

\subsubsection{Object-Level Type Formers}
Lastly, object-level type formers are straightforward. For types, we need
$\emb{\ev\,A\,\gamma} = \re\,A\,\gamma'$, and likewise for terms. Note that
$\ev$ and $\emb{\blank}$ preserve all structure, and $\re$ preserves all
structure on object-level types and terms. Hence, each object-level $A$ and $t$
case in $\blank\rel$ trivially follows from induction hypotheses.

This concludes the definition of the $\blank\rel$ interpretation.

\subsection{Soundness}

\begin{definition} First, we introduce shorthands for
external operations that can be obtained from $\blank\rel$.
\begin{alignat*}{5}
  & \text{For}\hspace{1em} \Gamma &&: \Con \hspace{1em}&&\text{we get}\hspace{1em}
         |\Gamma\rel| &&: \{\Delta : \Cono\} \to |\ev\,\Gamma|\,\Delta \to \Sub\,\emb{\Delta}\,\Gamma \to \Set_j\\
  & \text{For}\hspace{1em} A &&: \Ty_{0,j}\,\Gamma \hspace{1em}&&\text{we get}\hspace{1em} |A\rel| &&: |\Gamma\rel|\,\gamma\,\gamma' \to \emb{\ev\,A\,\gamma} = A[\gamma']\\
  & \text{For}\hspace{1em} t &&: \Tm_{0,j}\,\Gamma\,A \hspace{1em}&&\text{we get}\hspace{1em}|t\rel| &&: |\Gamma\rel|\,\gamma\,\gamma' \to \emb{\ev\,t\,\gamma} = t[\gamma']
\end{alignat*}
Since $\blank\rel$ was defined in $\hato$, we also know that the above are all
stable under object-theoretic substitution.
\end{definition}

\begin{theorem}[Soundness for generic contexts]
For each $\Gamma : \Cono$, we have ${\Gamma^P}\rel : |\Gamma\rel|\,\Gamma^P\,\id$.
\end{theorem}
\begin{proof}
We define ${\blank^P}\rel$ by induction on object-theoretic contexts.
${\emptycon^P}\rel : \top$ is defined trivially as $\tt$. For ${(\Gamma \ext
  A)^P}\rel$, we need $(\gamma\rel : |\Gamma|\rel\,(\Gamma^P[\p])\,\p) \times
(\emb{\q} = \q)$. We get ${\Gamma^P}\rel : |\Gamma\rel|\,\Gamma^P\,\id$. Because
of the naturality of $|\Gamma\rel|$, this can be weakened to ${\Gamma^P}\rel[\p]
: |\Gamma\rel|\,(\Gamma^P[\p])\,\p$. Also, $\emb{\q} = \q$ holds immediately.
\end{proof}

\begin{theorem}[Soundness of staging]
The open staging algorithm from Definition \ref{def:open-staging} is sound.
\end{theorem}
\begin{proof}
\mbox{}
\begin{itemize}
\item For $A : \Ty_{0,j}\,\emb{\Gamma}$, we have that $\emb{\Stage\,A} =
\emb{|\ev\,A|\,\Gamma^P}$ by the definition of $\Stage$, and moreover we have
$|A\rel|\,{\Gamma^P}\rel : \emb{|\ev\,A|\,\Gamma^P} = A[\id]$, hence $\emb{\Stage\,A} = A$.

\item For $t : \Tm_{0,j}\,\emb{\Gamma}\,A$, using $|t\rel|\,{\Gamma^P}\rel : \emb{|\ev\,t|\,\Gamma^P} = t[\id]$, we likewise
  have $\emb{\Stage\,t} = \emb{|\ev\,t|\,\Gamma^P} = t[\id] = t$.
\end{itemize}
\end{proof}

\begin{corollary}[Conservativity of 2LTT]\label{conservativity}
From the soundness and stability of staging, we get that $\emb{\blank}$ is
bijective on types and terms, hence $\Ty_{\mbbo j}\,\Gamma \simeq \Ty_{0,j}\,\emb{\Gamma}$ and
$\Tm_{\mbbo j}\,\Gamma\,A \simeq \Tm_{0,j}\,\emb{\Gamma}\,\emb{A}$.
\end{corollary}

\paragraph{Alternative presentations}
The above conservativity proof could be rephrased as an instance of more modern
techniques which let us implicitly handle stability under 2LTT substitutions as
well, not just $\mbbo$ substitutions. This amounts to working in a \emph{modal}
internal language, where modalities control the interaction of the different base
categories.
\begin{itemize}
\item Using \emph{synthetic Tait computability} \cite{sterlingthesis}, we work
      in the internal language of the glued category $\msf{Id}_{\hato}\downarrow
      \re$.
\item Using \emph{relative induction} \cite{bocquet2021relative}, we work in the modal type theory
      obtained from the dependent right adjoint functor $\re^* : \wh{\msf{2LTT}} \to \hato$,
      where $\wh{\msf{2LTT}}$ denotes presheaves over the syntactic category of 2LTT.
\end{itemize}
We do not use either of these in this paper, for the following reasons. First,
the author of the paper is not sufficiently familiar with the above
techniques. Second, the task at hand is not too technically difficult, so using
advanced techniques is not essential. Contrast e.g.\ normalization for cubical
type theories, which is not feasible to handle without the more synthetic
presentations \cite{cubicalnbe}.

\section{Intensional Analysis}\label{sec:intensional-analysis}

We briefly discuss intensional analysis in this section. This means analyzing
the internal structure of object-level terms, i.e.\ values of $\Lift\,A$. Disallowing
intensional analysis is a major simplification, which is sometimes enforced
in implementations, for example in MetaML \cite{metaml}.

If we have sufficiently expressive inductive types in the meta-level language,
it is possible to simply use inductive deeply embedded syntaxes, which can be
analyzed; our example in Section \ref{sec:staged-programming} for embedding a
simply typed lambda calculus is like this. However, native metaprogramming
features could be more concise and convenient than deep embeddings, similarly to
how staged programming is more convenient than code generation with deeply
embedded syntaxes.

In the following, we look at the semantics of intensional analysis in the
standard presheaf models. We only discuss basic semantics here, and leave
practical considerations to future work.

\begin{definition}
The \textbf{Yoneda embedding} is a functor from $\mbbo$ to $\hato$, sending
$\Gamma : \Cono$ to $\msf{y}\Gamma : \hCon$, such that $|\msf{y}\Gamma|\,\Delta
= \Subo\,\Delta\,\Gamma$. We say that $\Gamma : \hCon$ is \emph{representable}
if it is isomorphic to a Yoneda-embedded context.
\end{definition}

\begin{lemma}[Yoneda lemma]
We have $\hSub\,(\msf{y}\Gamma)\,\Delta \simeq |\Delta|\,\Gamma$
\cite[Section~III.2]{maclane98categories}. Also, we have
$\hTm_{1,j}\,(\msf{y}\Gamma)\,A \simeq |A|\,\id$, where $\id :
\Subo\,\Gamma\,\Gamma$.
\end{lemma}

The Yoneda lemma restricts possible dependencies on representable contexts.  For
example, consider the staging behavior of $t : \Tm_{1}\,(\emptycon \ext (x :
\Bool_0))\,\Bool_1$. Staging $t$ yields essentially a natural transformation
$\hSub\,(\ev(\emptycon \ext \Bool_0))\,\mbb{B}$, where $\mbb{B}$ is the
metatheoretic type with two elements, considered as a constant presheaf.  Now,
if $\ev(\emptycon\,\ext\, \Bool_0)$ is representable, then
$\hSub\,(\ev(\emptycon \ext \Bool_0))\,\mbb{B} \simeq \mbb{B}$, so $t$ can be
staged in at most two different ways.

Which contexts are representable? In the specific 2LTT in this paper,
$\emb{\Gamma}$ is always representable in the presheaf
model. $|\ev{\emb{\Gamma}}|\,\Delta$ contains lists of object-theoretic terms,
hence $\ev{\emb{\Gamma}} \simeq \yon \Gamma$.  This implies that intensional
analysis is \emph{not compatible} with the standard presheaf model, for our
2LTT.

Consider a very basic feature of intensional analysis, \emph{decidability of
definitional equality} for object-level Boolean expressions.
\[
  \msf{decEq} : (x\,y : \Lift\,\Bool_0) \to (x =_1 y) +_1 (x \neq_1 y)
\]
The Yoneda lemma says that $\msf{decEq}$ cannot actually decide
definitional equality. We expect that $\msf{decEq}$ lets us define non-constant
maps from $\Lift\,\Bool_0$ to $\Bool_1$, but the Yoneda lemma implies that we
only have two terms with distinct semantics in $\Tm_1\,(\emptycon \ext_0
\Bool_0)\,\Bool_1$.

A more direct way to show infeasibility of $\msf{decEq}$ in $\hato$ is to note
that definitional \emph{inequality} for object-level terms is not stable under
substitution, since inequal variables can be mapped to equal terms.

\paragraph{Stability under weakening only}
We can move to different 2LTT setups, where stability under substitution is not
required. We may have \emph{weakenings} only as morphisms in $\mbbo$.  A
weakening from $\Gamma$ to $\Delta$ describes how $\Gamma$ can be obtained by
inserting zero or more entries to $\Delta$. A wide range of intensional analysis
features are stable under weakenings. For example, $\msf{decEq}$ now holds in
$\hato$ whenever object-theoretic definitional equality is in fact decidable.
We also dodge the Yoneda lemma, since $\ev\,\emb{\Gamma}$ is not necessarily
representable anymore: $|\ev\,\emb{\Gamma}|\,\Delta$ is still a set of lists of
terms, but weakenings from $\Delta$ to $\Gamma$ are not lists of terms.

As a trade-off, if there is no notion of substitution in the specification of the
object theory, it is not possible to specify dependent or polymorphic types
there. We need a substitution operation to describe dependent elimination or
polymorphic instantiation.

However, stability under weakening is still sufficient for many practical use
cases, for example the monomorphization setup in Section
\ref{sec:monomorphization}. Here, we only have weakening as definitional
equality in the object theory, and we do not have $\beta\eta$-rules.

\paragraph{Closed modalities}\label{sec:closed-modalities}
Another option is to add a \emph{closed} or \emph{crisp} modality for modeling
closed object-level terms \cite{licata2018internal}. Since closed terms are not
affected by substitution, we are able to analyze them, and we can also recover
open terms by explicitly abstracting over contexts
\cite{DBLP:conf/lics/Hofmann99,DBLP:journals/corr/abs-2206-02831}. \emph{Function
pointers} in the style of the C programming language could be also an
interesting use case, since these must become closed after staging (as they
cannot capture closure environments).


\section{Related Work}\label{sec:related-work}

The work Annekov et al.\ on two-level type theory \cite{twolevel}, building on
\cite{capriotti2017models} and \cite{hts}, is the foremost inspiration in
the current work. Interestingly, the above works do mention metaprogramming as a
source of intuition for 2LTT, but they only briefly touch on this aspect, and in
\cite{twolevel} the main focus is on extensions of basic 2LTT which do not have
staging semantics anymore. Ibid.\ conjectures the strong conservativity of 2LTT
in Proposition 2.18, and suggests categorical gluing as proof strategy. In this
work, we do essentially that: our presheaf model and the logical relation model
could be merged into a single instance of gluing along the $\re$ morphism; this
merges staging and its soundness proof. We split this construction to two parts
in order to get a self-contained presentation of the staging algorithm.

\cite{DBLP:conf/lics/Hofmann99} models higher-order abstract syntax of simply
typed object languages in presheaves over various syntactic categories. The main
difference between HOAS and 2LTT is that the former represents object-level binders
using meta-level functions while the latter has a primitive notion of object-level
binding.

The multi-stage calculus in \cite{multi-stage-calculus} supports staging
operations and type dependencies. However, there are major differences to the
current work. First, ibid.\ does not have large elimination, hence it does not
support computing types by staging. Second, it does not support staged
compilation in the sense of Definition \ref{def:staging}; there is no judgment
which ensures that staging operations can be eliminated and a purely
``object-level'' program can be extracted from a multi-stage program. Rather, the system
supports \emph{runtime} code generation. There are several modal multi-stage
systems which support runtime code generation but not staged compilation, such
as those based on S4 modalities (e.g.\ \cite{DBLP:journals/jacm/DaviesP01}) or
contextual modalities (\cite{DBLP:journals/pacmpl/JangGMP22}).


Idris 1 supports compile-time partial evaluation, through \texttt{static}
annotations on function arguments \cite{scrap-your-inefficient-engine}. This can
be used for dependently typed code generation, but Idris 1 does not guarantee
that partial evaluation adequately progresses. For example, we can pass a neutral
value as static argument in Idris 1, while in 2LTT static function arguments
are always canonical during staging.

Our notation for quoting and splicing is borrowed from MetaML \cite{metaml}.  In
the following, we compare 2LTT to MetaML, MetaOCaml \cite{kiselyov14metaocaml}
and typed Template Haskell \cite{typed-th}.

2LTT explicitly tracks stages of types, in contrast to the mentioned
systems. There, we have type lifting (usually called $\msf{Code}$), quoting and
splicing, but lifting does not change the universe of types. We can write a
function with type $\Bool \to \msf{Code}\,\Bool$, and the two $\Bool$
occurrences refer to the exact same type.

This causes some operational confusion down the line, and additional
disambiguation is needed for stages of binders. For example, in typed Template
Haskell, top-level binders behave as runtime binders when used under a quote,
but behave as static binders otherwise. Given a top-level definition $b =
\True$, if we write a top-level definition $\msf{expr} = [||\,b\,||]$, then $b$
is a \emph{variable} in the staging output, but if we write $\msf{expr} =
\msf{if}\,b\,\msf{then}\,[||\,\True\,||]\,\msf{else}\,[||\,\False\,||]$, then
$b$ is considered to be a static value. In contrast, 2LTT does not distinguish
top-level and local binders, and in fact it has no syntactic or scope-based
restrictions; everything is enforced by typing.

2LTT supports staging for types and dependent types, while to our knowledge no
previous system for staged compilation does the same. It appears to the author
that the main missing component in previous systems is \emph{dependent types at
the meta level}, rather than staging operations; staging operations in 2LTT and
MetaOCaml are already quite similar. For example, any interesting operation on
the static-length vectors in Section \ref{sec:staged-programming} requires
meta-level dependent types, even if the object language is simply typed.

In MetaOCaml, there is a \emph{run} function for evaluating closed object-level
terms, however, this can fail during staging since closedness is not statically
ensured. In our setting, this could be reproduced in a safe way using a closed
modality; this is part of future work.


Additionally, 2LTT only supports two stages, while the other noted systems allow
countable stages. It remains future work to develop NLTT (N-level type theory),
but this extension does appear to be straightforward. In the specification of
NLTT, we would simply need to move from $i \in \{0,1\}$ to $i \in \mbb{N}$,
although in the semantics there could be more subtleties.

Many existing staged systems also support \emph{side effects} at compile time,
while our version of 2LTT does not. Here, general considerations about the
mixing of dependent types and effects should apply; see
e.g.\ \cite{fire-triangle}. However, it should be possible to use standard
effect systems such as monads at the meta level.

\section{Conclusions and Future Work}\label{sec:conclusions}

We developed the foundational metatheory of 2LTT as a system of
two-stage compilation. We view the current results as more like a starting point
to more thorough investigation into applications and extensions; in this paper
we only sketched these.

We emphasize that variants of 2LTT can serve as languages where staging
features can be practically implemented, and also as \emph{metatheories}, not
only for the formalization for staging techniques, but also more generally for
reasoning about constructions in object languages. This purpose is prominent in
the original 2LTT use case in synthetic homotopy theory. The meta-level language
in a 2LTT can be expressive enough to express general mathematics, and the
meta-level propositional equality $\blank\!=_1\!\blank$ can be used to prove
$\beta\eta$-equality of object-level types and terms.

For example, we could use object theories as \emph{shallowly embedded} target
languages of optimization techniques and program translations. In particular, we
may choose a low-level object theory without higher-order functions, and
formalize closure conversion as an interpretation of an embedded meta-level
syntax into the object theory.

There is much future work in adapting existing staging techniques to 2LTT
(which, as we mentioned, can double as formalization of said techniques), or
adding staging where it has not been available previously.
\begin{itemize}
\item \emph{Let-insertion techniques}. These allow more precise control over
  positions of insertion, or smarter selection of such positions. In 2LTT,
  meta-level continuation monads could be employed to adapt some of the features
  in \cite{DBLP:journals/jfp/KameyamaKS11}. The automatic let-floating feature
  of MetaOCaml \cite{DBLP:journals/corr/abs-2201-00495} requires strengthening
  of terms, which is not stable under substitution, but it is stable under
  weakening, so perhaps it could be implemented as an ``intensional analysis''
  feature.
\item \emph{Generic programming with dependent types}. There is substantial
  literature in this area
  (e.g.\ \cite{loh11generic,chapman2010gentle,diehl-thesis,ornaments}), but only
  in non-staged settings.  2LTT should immediately yield staging for these
  techniques, assuming that the meta level has sufficient type formers
  (e.g.\ induction-recursion). We could also try to adapt previous work on
  generic treatment of partially static data types
  \cite{DBLP:journals/pacmpl/YallopGK18}; fully internalizing this in 2LTT would
  also require dependent types.
\end{itemize}

There are also ways that 2LTT itself could be enhanced.
\begin{itemize}
\item \emph{More stages, stage polymorphism}. Currently, there is substantial
  code duplication if the object-level and meta-level languages are similar.
  Stage polymorphism could help reduce this. A simple setup could include three
  stages, the runtime one, the static one, and another one which supports
  polymorphism over the first two stages.
\item \emph{Modalities}. We mentioned the closed modality in Section
  \ref{sec:closed-modalities}.  More generally, many \emph{multimodal}
  \cite{gratzer20multimodal} type theories could plausibly support staging. 2LTT
  itself can be viewed as a very simple multimodal theory with $\Lift$ as the
  only modality, which is also degenerate (because it entails no structural
  restriction on contexts). We could support multiple object theories, with
  modalities for lifting them to the meta level or for representing translations
  between them. We could also have modalities for switching between stability
  under substitution and stability under weakening.
\item \emph{Intensional analysis}. We only touched on the most basic semantics
  of intensional analysis in Section \ref{sec:intensional-analysis}. It remains
  a substantial challenge to work out the practicalities. For good ergonomics,
  we would need something like a ``pattern matching'' operation on object-level
  terms, or some induction principle which is more flexible than the plain
  assumption of decidable object-level equality.
\end{itemize}

\section{Data Availability Statement}

A standalone prototype implementation is available as \cite{staged-demo}. It
implements elaboration and staging for a 2LTT. It includes a tutorial and code
examples that expand on the contents of Section \ref{sec:tour-of-2ltt}.

\interlinepenalty=10000
\bibliography{references}

\end{document}